\documentclass[aps,reprint,superscriptaddress,nofootinbib, showkeys]{revtex4-2}
\usepackage{float}
\makeatletter
\let\newfloat\newfloat@ltx
\makeatother

\usepackage[english]{babel}
\usepackage[utf8]{inputenc}
\usepackage{graphics}
\usepackage{selinput}
\usepackage[normalem]{ulem}
\usepackage[shortlabels]{enumitem}

\usepackage{braket}
\usepackage{amsthm}
\usepackage{amssymb}
\usepackage{mathtools}
\usepackage{physics}
\usepackage{xcolor}
\usepackage{graphicx}
\usepackage{adjustbox}
\usepackage{placeins}
\usepackage[T1]{fontenc}
\usepackage{lipsum}
\usepackage{csquotes}
\usepackage{bm}

\usepackage[linesnumbered,ruled,vlined]{algorithm2e}
\SetKwInput{kwInit}{Init}

\def\HC{\mathcal{H}}

\def\LC{\mathcal{L}}

\def\ad{^{\dagger}}

\newcommand{\fsnull}[1]{}
\newcommand{\old}[1]{}

\usepackage[toc,page]{appendix}
\usepackage[colorlinks=true,citecolor=blue,linkcolor=magenta]{hyperref}

\usepackage[utf8]{inputenc}
\usepackage{graphicx}
\usepackage{xcolor}
\usepackage{amsmath}
\usepackage{amsthm}
\usepackage{bm}
\usepackage{bbm}
\usepackage{comment}
\usepackage{appendix}
\usepackage{mathdots}
\usepackage{lipsum}
\usepackage{verbatim}
\usepackage{natbib}
\usepackage{nccmath}
\usepackage{amsfonts}

\newcommand{\dya}[1]{\ket{#1}\!\bra{#1}}

\newcommand{\PC}{\mathcal{P}}

\newcommand{\XC}{\mathcal{X}}

\renewcommand{\geq}{\geqslant}
\renewcommand{\leq}{\leqslant}

\newcommand{\bs}{\textsf{BS}}

\def\be{\begin{equation}}
\def\ee{\end{equation}}
\def\bs{\begin{split}}
\def\e{\end{split}}
\def\ba{\begin{eqnarray}}
\def\bea{\begin{eqnarray}}

\def\tea{\end{eqnarray}}
\def\ea{\end{eqnarray}}
\def\eea{\end{eqnarray}}

\def\SU{\text{SU}}

\def\U{\mathrm{U}}

\newtheorem{theorem}{Theorem}

\newtheorem{proposition}{Proposition}

\newtheorem{definition}{Definition}

\begin{document}
\title{Group Fourier filtering of quantum resources in quantum phase space} 

\author{Luke Coffman}
\thanks{lukecoffman@fas.harvard.edu}
\affiliation{Theoretical Division, Los Alamos National Laboratory, Los Alamos, NM 87545, USA}
\affiliation{Harvard Quantum Initiative, Harvard University, Cambridge, Massachusetts 02138, USA}

\author{N. L. Diaz}
\affiliation{Information Sciences, Los Alamos National Laboratory, Los Alamos, NM 87545, USA}
\affiliation{Center for Non-Linear Studies, Los Alamos National Laboratory, 87545 NM, USA}

\author{Mart\'{i}n Larocca}
\affiliation{Theoretical Division, Los Alamos National Laboratory, Los Alamos, NM 87545, USA}
\affiliation{Quantum Science Center, Oak Ridge, TN 37931, USA}

\author{Maria Schuld}
\affiliation{Xanadu, Toronto, ON, M5G 2C8, Canada}

\author{M. Cerezo}
\thanks{cerezo@lanl.gov}
\affiliation{Information Sciences, Los Alamos National Laboratory, Los Alamos, NM 87545, USA}
\affiliation{Quantum Science Center, Oak Ridge, TN 37931, USA}

\begin{abstract}
  Recently, it has been shown that group Fourier analysis of quantum states, i.e., decomposing them into the irreducible representations (irreps) of a symmetry group, enables new ways to characterize their resourcefulness. Given that quantum phase spaces (QPSs) provide an alternative description of quantum systems, and thus of the group’s representation, one may wonder how such harmonic analysis changes. In this work we show that for general compact Lie-group quantum resource theories (QRTs), the entire family of Stratonovich-Weyl quantum phase-space representations, characterized by the Cahill-Glauber parameter $s$, has a clear resource-theoretic and signal-processing meaning. Specifically, changing $s$ implements a continuously tunable group Fourier filter. In the examples of spin coherence, multipartite entanglement, and fermionic Gaussianity, this gives a low-pass representation at $s=-1$, an unfiltered spectrum at $s=0$, and a high-pass representation at $s=1$. Moreover, we show that the norms of the QRT’s free-state Fourier components completely determine this family of phase-space representations. Finally, we uncover an $s$-duality relating the nontrivial phase-space spectra of free states to the ensemble-averaged spectra of typical (Haar-random) states through a shift in $s$. Overall, our results provide a new interpretation of QPSs and establish them as a signal-processing framework for characterizing, filtering, and visualizing quantum resources.
\end{abstract}

\maketitle

\section{Introduction}

Understanding and characterizing the resourcefulness of quantum states is fundamental for quantum information sciences to analyze its usefulness within a given context. For instance, entanglement can enable beyond-classical communication and computations~\cite{nielsen2000quantum}, and non-stabilizerness promotes Clifford circuits to universal quantum computing~\cite{veitch2014resource,howard2017application,chitambar2019quantum, leone2022stabilizer}. Quantum resource theories (QRTs)~\cite{chitambar2019quantum,brandao2015reversible,diaz2025unified} provide a mathematical framework to quantify the resourcefulness of a state (within a given paradigm~\cite{deneris2025analyzing}), the type and the structure of the resource~\cite{dur2000three}, as well as how resources can be interconverted.

Given the intrinsic connection between QRTs and group theory, as the ``free'' resourceless operations of a given QRT are usually defined in terms of a representation $T$ of a group $\mathbb{G}$, tools from representation theory have become fundamental in resource quantification. These include the computation of projections of a state onto the elements of a Lie group (i.e., characteristic functions)~\cite{gu1985group,korbicz2006group,korbicz2008entanglement,marvian2013theory} or onto the elements of a Lie algebra or preferred subset of observables~\cite{barnum2003generalizations,barnum2004subsystem,klyachko2002coherent,delbourgo1977maximum, meyer2002global,brennen2003observable,beckey2021computable,schatzki2022hierarchy,balachandran2013entanglement,harshman2011observables,zanardi2001virtual,zanardi2004quantum,nha2006entanglement,alicki2009quantum,viola2010entanglement,derkacz2011entanglement,gigena2015entanglement,benatti2016entanglement,regula2017convex,sindici2018simple,gigena2021many,guaita2021generalization,ahmad2022quantum}, as well as the analysis of group-invariant quantities~\cite{grassl1998computing,barnum2001monotones,miyake2003classification,leifer2004measuring,mandilara2006quantum,klyachko2007dynamical,oszmaniec2013universal,bravyi2019simulation,larocca2022group,meyer2023exploiting,nguyen2022atheory,skolik2022equivariant}. In  particular, the recent work of~\cite{bermejo2025characterizing} proposed performing group harmonic analysis of quantum states, i.e., decomposing them into the irreducible representations (irreps) of the different vector spaces admitting a representation of the group of free operations (for a different characterization using group Fourier analysis see~\cite{marvian2014modes}). For the QRTs studied there, free states tend to have more weight in lower-dimensional irreps, whereas resourceful states tend to populate higher-dimensional sectors. This motivates a comparison with harmonic analysis. While Ref.~\cite{bermejo2025characterizing} studied resourcefulness based on the standard description of quantum mechanics (i.e., a Hilbert space $\HC$, where quantum states are described by normalized positive definite matrices $\rho$ belonging to the space of linear operations $\LC(\HC)$), one may wonder if such investigation can be repeated across alternative descriptions of quantum mechanical systems, such as those obtained via quantum phase spaces (QPSs).

Indeed, we recall that QPSs provide an alternative
description of quantum mechanics in which states and operators are represented by functions on a classical-like manifold, equipped with quasi-probability distributions such as the Wigner, Husimi $Q$, or Glauber--Sudarshan $P$ functions~\cite{wigner1932quantum,cahill1969ordered}. For spin systems, these representations connect the multipole components of states and operators to spherical harmonics on the phase space~\cite{agarwal1981relation,varilly1989moyal}. In the general group-theoretic setting~\cite{brif1999phase}, the Stratonovich-Weyl (SW) correspondence associates to each operator $A\in \LC(\HC)$ a function $F_A(\Omega,s)$, real when $A$ is Hermitian, parametrized by the Cahill-Glauber parameter $s$, on a homogeneous space $\XC=\mathbb{G}/\mathbb{K}$. Here $\mathbb{G}$ is a symmetry group with unitary representation $T$ and $\mathbb{K}$ is the stabilizer group of a distinguished reference state (typically the highest-weight, or coherent, state)~\cite{perelomov1977generalized,zhang1990coherent,brif1999phase}. The harmonic expansion of the SW kernel relates the operator and phase-space components, with its coefficients determining the weight assigned to each component~\cite{brif1999phase}. QPSs have recently gained considerable attention in modern quantum computation and information~\cite{sanchez2025phase} as they also enable new viewpoints for quantum machine learning~\cite{heightman2025quantum,chabaud2024phase} and the study of classical simulability for quantum systems~\cite{ipek2025phase}.

From the previous, it is clear that QPSs and QRTs  have a striking similarity: they are built from the exact same ingredients, a group of free operations and a reference free state. This connection has motivated a series of works exploring whether resourcefulness can be detected directly from
phase space representations~\cite{ferrie2011quasi}. In continuous-variable settings, Wigner negativity and related quasi-probability features have been identified as indicators of nonclassicality and as necessary resources for quantum computational advantage, leading to explicit resource theories of Wigner negativity and non-Gaussianity~\cite{veitch2014resource,albarelli2018resource,rahimi2013measurement}. More broadly, Wigner negativity has been linked to contextuality and other operational notions of quantum advantage~\cite{booth2022contextuality}. In finite dimensions, the QPS representations such as the Wigner function have been used to characterize nonclassicality and entanglement ~\cite{tilma2016wigner,veitch2012negative}, although for phase-spaces of finite dimensional systems negativity can already appear for free states and does not align cleanly with standard resourcefulness quantifiers in those QRTs ~\cite{davis2021wigner,ferrie2008frame, zhu2016quasiprobability}. 

The shared ingredients of QRTs and QPSs raise a question: \textit{how does the choice of phase-space representation reshape the group Fourier profile used to characterize quantum resources?} We address this question by developing a general construction of phase-space filters from the irrep purities of a reference free state. This construction connects resource profiles in operator space and phase space, and makes the dependence on the QRT explicit. 

Our main contributions are as follows (see Fig.~\ref{fig:schematic}). First, within the SW family and the multiplicity-free, injective setting specified below, we show that the reference free-state GFD purities determine the kernel coefficients. This gives a common construction across compact Lie-group QRTs. Second, we establish the exact relation between the operator-space and phase-space irrep purities of any state, with multiplier $\tau_\lambda^{-s}$. Together, these results identify the SW representations as group Fourier filters whose spectrum is fixed by the QRT's reference free state. In all three examples considered, $\tau_\lambda$ decreases with the natural Casimir frequency, indicating a veritable low-pass and high-pass filter. Third, we relate the nontrivial-irrep purities of free states to those of Haar-random pure states, averaged over the Haar ensemble, through a shift in $s$. 

\section{Background}

\begin{figure}
    \centering
    \includegraphics[width=1\columnwidth]{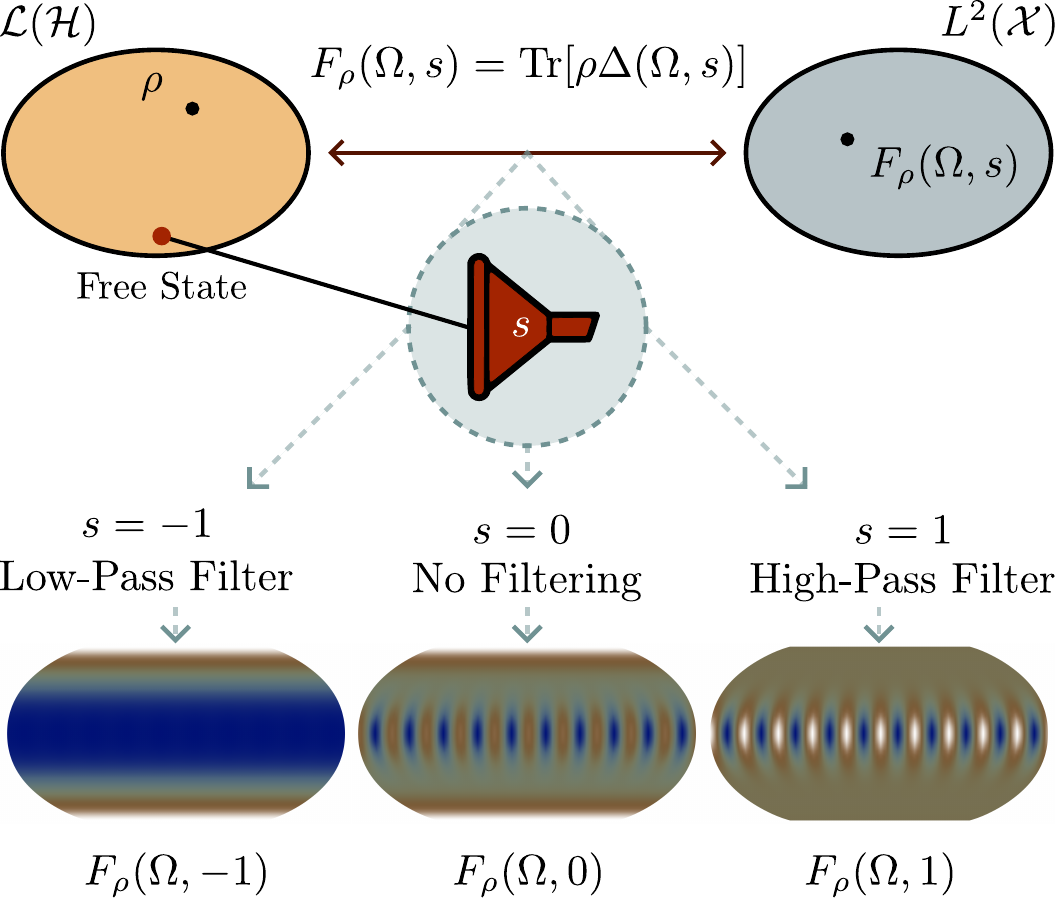}
    \caption{\textbf{Schematic representation of our main results.} We show that the SW kernels $\Delta(\Omega,s)$ serve as a signal-processing-like Fourier filter for the irrep components of quantum operators, with the Cahill-Glauber parameter $s$ interpolating between low-pass ($s=-1$) and high-pass ($s=1$) actions. Within the multiplicity-free SW family considered, the irrep weights are determined by the reference free state. The low-pass and high-pass interpretation applies to the examples discussed here. The illustrated QPS corresponds to the Robinson projection of the sphere $S^2$ for a Greenberger–Horne–Zeilinger (GHZ) state in the spin coherence QRT. See below for more details. }
    \label{fig:schematic}
\end{figure}

In this section we present several important concepts that will be used throughout the work. First, we will review well-known results of how, under the adjoint action of a unitary group, the space of linear operators of a quantum's Hilbert space decomposes into subspaces called irreducible representations (irreps). Then, we will review the results of~\cite{bermejo2025characterizing} where we defined projections and purities onto the irreps of operator space. Next, we briefly present the basic ingredients for QRTs~\cite{chitambar2019quantum}. Finally, we will recall the SW correspondence to build QPSs, and present as well the irrep decomposition of functions over phase space. In the next section we will see how all these concepts are intimately related.

\subsection{Group Fourier decomposition of the operator space}

Let $\HC=\mathbb{C}^d$ be a $d$-dimensional Hilbert space, and let $T$ be a unitary irreducible representation of a compact Lie group $\mathbb{G}$ with associated semisimple Lie algebra $\mathfrak{g}$. We use $dT$ for its Lie algebra representation. Throughout the scalar SW construction we assume that the conjugation action on $\LC(\HC)$ is multiplicity-free, that its irreps are stable under Hermitian conjugation.

We are interested in characterizing the action of the unitaries in $T(\mathbb{G})$ over elements of the space of linear operators $\LC(\HC)$, i.e., 
\begin{equation}\label{eq:action}
    T(g)AT\ad(g)\,, \text{ for $g\in\mathbb{G}$ and $A\in\LC(\HC)$}.
\end{equation}
For this purpose, we recall the following result~\cite{serre1977linear,fulton1991representation}. 
\begin{definition}[Irrep decomposition of $\LC(\HC)$]\label{def:irreps-LH}
   The action of $T(\mathbb{G})$ induces a decomposition of $\LC(\HC)$ into irreps as
\begin{equation}
  \LC(\HC)
  =\HC\otimes \HC^{*}\cong\bigoplus_{\lambda}V_{\lambda}\,,
\end{equation}
where the index $\lambda$ runs over the finitely many distinct irreps. For each $V_\lambda \subseteq \LC(\HC)$ we define a Hermitian orthonormal basis
\begin{equation}
  V_\lambda={\rm span}_{\mathbb{C}}\{D_j^{\lambda}\}_{j=1}^{d_\lambda}, 
  \end{equation}
  with
  \begin{equation}
  \left\langle D_j^{\lambda},D_{j'}^{\lambda'}\right\rangle_{\LC}
  =\delta_{\lambda\lambda'}\,\delta_{jj'}\,,
\end{equation}
where $d_\lambda=\dim(V_\lambda)$, and where $\langle A,B\rangle_{\LC}=\Tr[A\ad B]$ denotes the Hilbert-Schmidt inner product for the operators $A,B\in\LC(\HC)$. 
\end{definition}
Definition~\ref{def:irreps-LH} implies that the space of linear operators decomposes into invariant subspaces under the action of the group. As such, given some $A\in V_\lambda$, then $\forall g\in\mathbb{G}$ and one has $T(g)AT\ad(g)\in V_\lambda$. Thus, we can simplify the study of Eq.~\eqref{eq:action} by analyzing how the group acts on each irrep.

Indeed, given some unitary $T(g)$, which we assume for simplicity to be of the form  $T(g)=e^{\theta\,dT(h)}$ for $h\in\mathfrak{g}$, we can explicitly build its action over $V_\lambda$ by studying how it acts over a basis of such a subspace, i.e., by the transformation $T(g)D_j^{\lambda} T\ad(g)$. This is achieved via the so-called adjoint and Adjoint representations on the $\lambda$-th irrep, defined as follows.  
\begin{definition}[adjoint representation of $\mathfrak{g}$ in $V_\lambda$]\label{def:adjoint-rep}
    Given an operator $h$ in $\mathfrak{g}$, the adjoint representation in the irrep $V_\lambda$ is the map $\Phi_{\rm ad}^\lambda:\mathfrak{g}\rightarrow \mathbb{R}^{d_\lambda\times d_\lambda}$ and is defined by
    \begin{equation}
        \left(\Phi_{\rm ad}^\lambda(h)\right)_{jj'}=\Tr\left[iD_{j'}^{\lambda}\left[dT(h),i D_j^{\lambda}\right]\right]\,,
    \end{equation}
    where $dT$ is the Lie algebra homomorphism induced by $T$. 
\end{definition}
The adjoint representation of the Lie algebra induces, through the exponential map, the group action on $V_\lambda$. 
\begin{definition}[Adjoint representation of $\mathbb{G}$ in $V_\lambda$]\label{def:Adjoint-rep}
    Given a group element $g=e^{\theta h}$ of $\mathbb{G}$, where $h\in \mathfrak{g}$,  $\Phi_{\rm ad}^\lambda$ induces the unitary Adjoint representation  $\Phi_{\rm Ad}^\lambda:\mathbb{G}\rightarrow \mathbb{U}\left(\mathbb{R}^{d_\lambda}\right)$ given by
    \begin{equation}
        \Phi_{\rm Ad}^\lambda\left(g=e^{\theta \,h}\right)=e^{\theta \,\Phi_{\rm ad}^\lambda(h)}\,.
    \end{equation}
\end{definition}
Note that, due to linearity of the adjoint representation the Adjoint representation of some $ \Phi_{\rm Ad}^\lambda\left(e^{\sum _i\theta_i \,h_i}\right)$ is simply given by $e^{\sum _i\theta_i \,\Phi_{\rm ad}^\lambda(h_i)}$. Importantly, Definitions~\ref{def:adjoint-rep} and~\ref{def:Adjoint-rep} imply
\begin{equation}\label{eq:basistransf}
    T(g)D_j^{\lambda} T\ad(g)= \sum_{j'}\left(\Phi_{\rm Ad}^\lambda\left(g\right)\right)_{jj'} D_{j'}^{\lambda}\,,
\end{equation}
which determines how $\mathbb{G}$ acts on each irrep of $\LC(\HC)$. Then, via linearity we can fully characterize the transformation of Eq.~\eqref{eq:action}.

At this point, we find it important to note that the notions of adjoint and Adjoint representations may appear
technical at first glance, but they are exactly what allows us to
organize operator space $\LC(\HC)$ into symmetry sectors. In particular,
the adjoint action of $\mathbb{G}$ on $\LC(\HC)$ is the one under which both the group Fourier decomposition (GFD) purities and the SW kernels decompose into irreps. Making this action explicit now will let us track, later on, how each irrep
contributes to resource measures in $\LC(\HC)$ and to phase-space
functions.

Next, we will use  Definition~\ref{def:irreps-LH} to define the group Fourier decomposition (GFD) of an operator  $A\in\LC(\HC)$ as the collection of its projections into each irrep. Moreover, we dub the Hilbert-Schmidt inner product of such projections as the GFD purities. 
\begin{definition}[GFD and purities of operators in $\LC(\HC)$]\label{def:GFD-purities-LH}
    Given a linear operator $A$ in $\LC(\HC)$, its projection onto the $\lambda$-th irrep is given by 
    \begin{equation}
        A_{\lambda}=\sum_{j=1}^{d_\lambda} \left\langle D_j^{\lambda},A\right\rangle_{\LC} D_j^{\lambda}\,.
    \end{equation}
 The GFD purity on the $\lambda$-th irrep  is therefore obtained as
\begin{align}
    \PC_\lambda(A)&=\langle A_{\lambda},A_{\lambda}\rangle_{\LC}\nonumber\\
    &=\sum_{j=1}^{d_\lambda}|\left\langle D_j^{\lambda},A\right\rangle_{\LC}|^2\nonumber\\
    &=\sum_{j=1}^{d_\lambda}| \Tr[D_j^{\lambda}A]|^2\,.\label{eq:purity}
\end{align}
\end{definition}
One can readily verify that since $A=\sum_\lambda A_{\lambda}$, then $\sum_\lambda \PC_\lambda(A)=\langle A,A\rangle_{\LC}$, indicating that if $A$ is normalized, then the collection of purities $\PC_\lambda(A)$ forms a probability distribution. Moreover, a straightforward calculation reveals that 
\begin{equation}\label{eq:purity-inv}
    \PC_\lambda(T(g)AT\ad(g))=\PC_\lambda(A)\,,\quad \forall g\in\mathbb{G}\,,
\end{equation}
meaning that the GFD purities are invariant under conjugation by the unitaries in $T(\mathbb{G})$. 

As we discuss in the next section, the GFD purities have been shown to be relevant within the context of QRTs~\cite{bermejo2025characterizing}. Moreover, we will later see in this paper that they are also central to understand QPSs as group Fourier processing filters.

\subsection{Quantum resource theories}

Here we will define a QRT~\cite{chitambar2019quantum,brandao2015reversible} over $\HC$ in terms of two basic ingredients: a set of operations which are assumed to be easy and free to implement, and a reference state which is assumed  to be resourceless. Naturally, all the states that can be obtained by evolving such reference state  from free operations are also considered to be resourceless. For our purposes, we will consider QRTs with unitary, i.e., resource-preserving, free operations.
\begin{definition}[Free operations and states in Lie-algebraic QRTs]\label{def:QRTs}
We define the free operations of a QRT as given by the irreducible representation $T$  of a compact group $\mathbb{G}$. Then, we take the highest weight state of $T(\mathfrak{g})$ as the reference state, which we denote as $\ket{{\rm hw}}$, such that the set of all free states is the orbit $\{T(g)\ket{{\rm hw}}\,, g\in \mathbb{G}\}$.  
\end{definition}
Note that in Definition~\ref{def:QRTs}, we have chosen the reference state to be the highest weight of the algebra.  While a priori this choice seems arbitrary, it turns out that most QRTs are implicitly picking this state~\cite{bermejo2025characterizing}. Indeed, such choice makes the set of free states correspond to the so-called generalized coherent states~\cite{barnum2003generalizations,barnum2004subsystem,perelomov1977generalized,gilmore1974properties,zhang1990coherent}. This also guarantees that free states minimize generalized uncertainty relations~\cite{delbourgo1977maximum,barnum2004subsystem} and that their orbits possess underlying K\"ahler structures~\cite{kostant1982symplectic}, somewhat making them the ``most classical'' set of states.  

A central problem within QRTs is the identification, quantification and characterization of the resource present in a state $\ket{\psi}$ which is not free. For instance, the coherent state fidelity~\cite{regula2017convex,bravyi2019simulation,reardon2024fermionic}
\begin{align}\label{eq:fidelity}
S(\ket{\psi})&=\max_{g\in\mathbb{G}}\left\langle\dya{\psi},T(g)\dya{{\rm hw}}T\ad(g)\right\rangle_{\LC}\nonumber\\
&=\max_{g\in\mathbb{G}}|\langle {\rm hw} |T(g)|\psi\rangle|^2\,,
\end{align}
simply compares $\ket{\psi}$ with all coherent states, and is equal to one if $\ket{\psi}$ can be obtained by evolving the reference state with some free operation. In addition, there exist irreps in $\LC(\HC)$  whose purities serve as resource quantifiers~\cite{barnum2003generalizations,barnum2004subsystem}. For instance, when $\mathfrak{g}$ is simple  one can verify that there is some $\lambda'$ for which $\LC_{{\lambda}'}={\rm span}_{\mathbb{C}}\{dT(\mathfrak{g})\}$ and
\begin{equation}\label{eq:g-purity}
    \PC_{{\lambda}'}(\dya{\psi}) \text{ is maximized iff } \ket{\psi}=T(g)\ket{{\rm hw}}
\end{equation}
for some $g\in \mathbb{G}$. 

While Eqs.~\eqref{eq:fidelity}--\eqref{eq:g-purity} can determine whether a state is resourceful or not, the fact that they are single scalar quantities greatly limits the amount of information one can extract from them.  To mitigate this issue, the recent work of Ref.~\cite{bermejo2025characterizing} showed that the whole collection of GFD purities carries useful information regarding the structure of the resourcefulness in a given state. Notably, it was shown that free states tend to be supported in lower-dimensional irreps, whereas resourceful states have support in more, and higher-dimensional ones. Such behavior closely mimics that which is observed in classical harmonic analysis.

\subsection{Quantum phase spaces}\label{sec:normalization}

We begin by describing the basic ingredients of QPSs. Notably,  their two basic ingredients are exactly the same as those that define QRTs with resource-preserving operations: a unitary representation $T$ of a  group $\mathbb{G}$, and a reference state. Here we will focus on the case when $\mathbb{G}$ is a compact Lie group, and the reference state is the highest weight $\ket{{\rm hw}}$ of the associated Lie algebra.

As in QRTs a central object of QPSs is the orbit $\{T(g)\ket{{\rm hw}}\,, g\in \mathbb{G}\}$, but its study is refined by defining the subgroup $\mathbb{K}\subseteq \mathbb{G}$ which  stabilizes the reference state, i.e.,  $\mathbb{K}$ is composed of  all the group elements $k$ such that $T(k)$ leave $\ket{{\rm hw}}$ invariant up to a phase
factor.  Specifically, one has
\begin{equation}
    T(k)\ket{{\rm hw}}=e^{i\phi(k)}\ket{{\rm hw}}\,,\quad \forall k \in \mathbb{K}\,.
\end{equation}
Indeed, for all group elements $g\in \mathbb{G}$ there exists a decomposition in terms of left cosets as
\begin{equation}
    g=\Omega k\,, \text{ where } k\in \mathbb{K}\,, \text{ and } \Omega\in \mathbb{G}/\mathbb{K}.
\end{equation} 
Since any pair of elements, $g=\Omega k$ and $g'=\Omega k'$, will lead to the same state (up to a global phase), we identify free states with points $\Omega$ in $\XC:=\mathbb{G}/\mathbb{K}$ such that
\begin{equation}\label{eq:Omega-state}
  \ket{\Omega}:= T(\Omega)\ket{\rm hw}=T(g)\ket{\rm hw}
\end{equation}
where $g=\Omega k$ for some $k\in \mathbb{K}$. From here, one defines the QPS as the homogeneous space $\XC:=\mathbb{G}/\mathbb{K}$:
\begin{definition}[Quantum phase space]\label{def:QPSs}
    Given a representation $T$ of a group $\mathbb{G}$, a reference state $\ket{{\rm hw}}$ and its associated stabilizer group $\mathbb{K}\subseteq \mathbb{G}$, we define the QPS as $\XC=\mathbb{G}/\mathbb{K}$.
\end{definition}

We use the uniform invariant probability measure throughout, so that $\mu(\XC)=1$. Irreducibility then gives the coherent-state resolution
\begin{equation}\label{eq:completness}
\int_{\XC}d\mu(\Omega)\dya{\Omega}=\frac{I}{d}.
\end{equation}
Our normalization of the SW family keeps both the Husimi kernel $\Delta(\Omega,-1)=\dya{\Omega}$ and an isometric Wigner map at $s=0$. Consequently, standardization depends on $s$, as specified below. Explicitly, our kernels are related to the customary trace-one kernels by $\Delta(\Omega,s)=d^{(s+1)/2}\Delta_{\rm tr=1}(\Omega,s)$. This global rescaling changes absolute SW correspondence norms but not their relative harmonic weights.

Starting from Definition~\ref{def:QPSs} we can follow the SW correspondence~\cite{brif1999phase}. The linear injective mappings between the space of operators on the Hilbert
space $\LC(\HC)$ and that of square-integrable functions on the QPS $\XC$, denoted as  $L^2(\XC)$, satisfy the following properties:

\begin{definition}[Properties of phase space functions]\label{def:QPS-function}
  Let $A\in\LC(\HC)$ and $\Omega\in\XC$. Then, any linear  phase space function  $F_{A}(\Omega,s)$ parametrized by the real-valued Cahill-Glauber parameter $s$ satisfies the following properties 
     \begin{enumerate}
  \item[(i)] Reality: $F_{A\ad}(\Omega,s)=F_A^*(\Omega,s)$,
  \item[(ii)] Standardization: $\displaystyle\int_{\XC}\!d\mu(\Omega)F_A(\Omega,s)\!=\!d^{(s-1)/2}\Tr[A]$,
  \item[(iii)] Covariance: $F_{T\ad(g) AT(g)}(\Omega,s)=F_A(g\cdot \Omega,s)$,
  \item[(iv)] Tracing: 
    $\displaystyle\int_{\XC} d\mu(\Omega)F_{A^\dagger}(\Omega,s)F_B(\Omega,-s)\!=\!\Tr[A\ad B]$.
\end{enumerate}
\end{definition}

The tracing property pairs the dual maps at $s$ and $-s$ and can be expressed as
\begin{equation}\label{eq:sms-inner}
  \langle F_A(\cdot,s),F_B(\cdot,-s)\rangle_{L^2}=\langle A,B\rangle_{\LC}\,,
\end{equation}
 with 
 \begin{equation}
   \langle F,G\rangle_{L^2}=\displaystyle\int_{\XC}d\mu(\Omega)F^*(\Omega)G(\Omega)\label{eq:inner-prody-L2}
 \end{equation}
Next, we remark that Riesz's representation theorem states that for any $\Omega$ there always exists an operator $\Delta(\Omega,s)\in\LC(\HC)$, called the SW kernel, such that 
\begin{equation}\label{eq:PS-function}
  F_A(\Omega,s)=\langle \Delta(\Omega,s),A\rangle_{\LC}=\Tr[\Delta(\Omega,s)^\dagger A]\,.
\end{equation}
Since the tracing property holds $\forall B\in\LC(\HC)$, and thanks to injectivity, then we also have
\begin{equation}\label{eq:QPS-to-HC}
    A=\displaystyle\int_{\XC} d\mu(\Omega)F_A(\Omega,s)\Delta(\Omega,-s)\,.
\end{equation}
The properties of phase space functions enumerated in Definition~\ref{def:QPS-function} induce properties on the SW kernels, as explained below.

\begin{definition}[Properties of the SW kernels]\label{def:SW-kernel}
Let $\Delta(\Omega,s)$ be the kernel associated with the SW phase space mapping parametrized by $s$. Then, $\Delta(\Omega,s)$ satisfies the following properties 
\begin{enumerate}
  \item[(i)] Hermiticity: $\Delta(\Omega,s)^\dagger=\Delta(\Omega,s)$.
  \item[(ii)] Resolution: $\displaystyle\int_{\XC}d\mu(\Omega)\,\Delta(\Omega,s)=d^{(s-1)/2}I$, with $\Tr[\Delta(\Omega,s)]=d^{(s+1)/2}$.
  \item[(iii)] Covariance:     $\Delta\bigl(g\!\cdot\!\Omega,s\bigr)
     =T(g)\,\Delta(\Omega,s)\,T\ad(g).$
\end{enumerate}
\end{definition}
Because QPSs are not unique (e.g., different $s$ parameters in the phase space functions lead to different representations of an operator $A$ in $L^2(\XC)$) we can relate different representations as
\begin{equation}\label{eq:map-Fs}
F_A(\Omega,s)=\displaystyle\int_{\XC} d\mu(\Omega')K_{s,s'}(\Omega,\Omega')F_A(\Omega',s')\,,
\end{equation}
where
\begin{equation}\label{eq:kernel-Deltas}
    K_{s,s'}(\Omega,\Omega')=\Tr[\Delta(\Omega,s)\Delta(\Omega',-s')]\,.
\end{equation}
Indeed, Eq.~\eqref{eq:kernel-Deltas} directly implies the relation
\begin{equation}\label{eq:Delta-to-Delta}
    \Delta(\Omega,s)=\displaystyle\int_{\XC} d\mu(\Omega') K_{s,s'}(\Omega,\Omega')\Delta(\Omega',s')\,,
\end{equation}
where we can see that setting $s=s'$ indicates that $K_{s,s}(\Omega,\Omega')$ is the reproducing kernel on the finite-dimensional image of the SW correspondence in $L^2(\XC)$.

Since the homogeneous space $\XC$ was originally constructed from the Lie group $\mathbb{G}$, it naturally admits an action of $\mathbb{G}$ via translations, and hence so does phase space $L^2(\XC)$. Therefore, just as we can decompose $\LC(\HC)$ into irreps, we can do the same for $L^2(\XC)$ via the Peter-Weyl theorem~\cite{folland2016course}:
\begin{theorem}[Irrep decomposition of $L^2(\XC)$]\label{def:irreps-XC}
   The action of $\mathbb{G}$ induces a decomposition of $L^2(\XC)$ into irreps  as
\begin{equation}\label{eq:irrep-L2X}
  L^2(\XC)\cong\bigoplus_{\sigma}W_{\sigma}\,,
\end{equation}
where the index $\sigma$ runs over all infinitely many irreps and multiplicity label. For each $W_\sigma$ we define an orthonormal basis of harmonic functions
\begin{equation}
  W_\sigma={\rm span}_{\mathbb{C}}\{Y_j^{\sigma}(\Omega)\}_{j=1}^{d_\sigma},
  \end{equation}
  with
  \begin{equation}
  \left\langle Y_j^{\sigma},Y_{j'}^{\sigma'}\right\rangle_{L^2}
  =\delta_{\sigma\sigma'}\,\delta_{jj'}\,,\label{eq:ortogonal-Y}
\end{equation}
where $d_\sigma=\dim(W_\sigma)$, and where $\langle A,B\rangle_{L^2}$ denotes the inner product over $L^2(\XC)$ defined in~\eqref{eq:inner-prody-L2}. 
\end{theorem}
From here, we can define projections and purities onto the irreps of $L^2(\XC)$.
\begin{definition}[GFD and purities of operators in $L^2(\XC)$]\label{def:GFD-purities-LX}
    Given a square integrable function $F(\Omega)$ over $\XC$, its projection onto the $\sigma$-th irrep is given by 
    \begin{equation}
        F^{\sigma}(\Omega)=\sum_{j=1}^{d_\sigma} \left\langle Y_j^{\sigma},F\right\rangle_{L^2} Y_j^{\sigma}(\Omega)\,.
    \end{equation}
The GFD purity on the $\sigma$-th irrep  is therefore obtained as
\begin{align}
    \widetilde{\PC}_\sigma(F)&=\langle F^{\sigma},F^{\sigma}\rangle_{L^2}\nonumber\\
    &=\sum_{j=1}^{d_\sigma}\left|\left\langle Y_j^{\sigma},F\right\rangle_{L^2}\right|^2\nonumber\\
    &=\sum_{j=1}^{d_\sigma}\left|\int_{\XC}d\mu(\Omega)(Y_j^{\sigma})^*(\Omega)F(\Omega)\right|^2\,.
\end{align}
\end{definition}
Here we note that the SW correspondence is a covariant isomorphism onto its image. It is an isometry for the ordinary $L^2$ inner product at $s=0$; at general $s$, Eq.~\eqref{eq:sms-inner} instead pairs the maps at $s$ and $-s$. Its harmonic image therefore contains matched copies of the operator irreps. We choose real harmonics corresponding to the Hermitian bases above, transforming as

\begin{equation}\label{eq:transform-Y}
    Y_j^{\lambda}(g^{-1}\cdot\Omega)=\sum_{j'}\left(\Phi_{\rm Ad}^\lambda\left(g\right)\right)_{jj'}Y_{j'}^{\lambda}(\Omega)\,.
\end{equation}

We adopt the one-parameter SW ansatz of Ref.~\cite{brif1999phase} so that the kernels take the form
\begin{equation}\label{eq:SW-kernel}
   \Delta(\Omega,s)=\sum_\lambda\sum_{j=1}^{d_\lambda}\tau_\lambda^{-s/2}Y^\lambda_j(\Omega)D_j^\lambda\,,
\end{equation}
where $\lambda$ runs over the distinct irreps in Definition~\ref{def:irreps-LH}. Injectivity requires $\tau_\lambda>0$ on all these sectors. Here $\tau_\lambda$ is defined by 
\begin{equation}\label{eq:tau0}
    \tau_\lambda^{-1/2}\bra{\Omega}D^\lambda_j \ket{\Omega}=Y^\lambda_j(\Omega)\,.
\end{equation}
In particular, $\tau_\lambda$ is independent of $j$ and $\Omega$. One can evaluate Eq.~\eqref{eq:tau0} at the reference coset $\Omega=e$ using a basis element with nonzero coherent-state expectation. Such an element lies in the weight-zero subspace of $V_\lambda$. Equation~\eqref{eq:tau0} implies that we can always express
\begin{equation}
    D^\lambda_j=\tau_\lambda^{-1/2}\displaystyle\int_{\XC} d\mu(\Omega)Y^\lambda_j(\Omega)\dya{\Omega}\,,
\end{equation}
and concomitantly 
\begin{align}
    \tau_\lambda^{-s/2}&=\left\langle Y_j^\lambda,\left\langle D_j^\lambda,\Delta(\cdot,s)\right\rangle_\LC\right\rangle_{L^2}\nonumber\\
    &=    \left\langle D_j^\lambda,\left\langle Y_j^\lambda,\Delta(\cdot,s)\right\rangle_{L^2}\right\rangle_\LC\,.\nonumber
\end{align}
Note that combining Eqs.~\eqref{eq:transform-Y} and~\eqref{eq:SW-kernel} directly leads to the covariance property of the SW kernels.

To finish, we note that some SW kernels are of particular interest. For instance, taking $s=-1$ one recovers the so-called Husimi Q function 
\begin{equation}\label{eq:Husimi}
    \Delta(\Omega,-1)=\dya{\Omega}\,,
\end{equation}
The positive operator-valued measure is $d\,\dya{\Omega}\,d\mu(\Omega)$. The kernels at other values of $s$ need not be positive and are understood as covariant operator-valued kernels, not as POVM densities. For $s=0$, one obtains the self-dual Wigner kernel
\begin{equation}\label{eq:Wigner}
    \Delta(\Omega,0)=\sum_\lambda\sum_{j=1}^{d_\lambda}Y^\lambda_j(\Omega)D_j^\lambda\,.
\end{equation}

To finish, we refer the reader to the Appendix where we present a table summarizing all main concepts of the Hilbert space and QPS formulation of quantum mechanics.

\section{Results}
As previously mentioned, in Ref.~\cite{bermejo2025characterizing} it was shown how the GFD in  $\LC(\HC)$ of a pure state $\rho$ can characterize its resourcefulness within the context of a QRT (see Definition~\ref{def:QRTs}). One can then wonder if mapping $\rho$ to the QPS $\XC$, and then performing its GFD decomposition, will lead to a different characterization. That is, we want to understand the differences and similarities which arise from projecting $\rho$ onto the irreps of $\LC(\HC)$ in Definition~\ref{def:irreps-LH}, versus projecting $F_{\rho}(\Omega,s)$ onto the irreps of $L^2(\XC)$ in Theorem~\ref{def:irreps-XC} and Definition~\ref{def:GFD-purities-LX}.

\subsection{Explicit formula for $\tau_\lambda$, and its connection to the GFD purity of free states}

The kernel coefficients can be expressed directly through the irrep purities of the reference free state. This identifies the state-dependent input to the chosen SW family and will let us compare its filters across QRTs.

\begin{proposition}\label{prop:tau}
The coefficient $\tau_\lambda$ is related to the GFD purity in $\LC(\HC)$ of the free highest weight state $\ket{{\rm hw}}$ as
    \begin{equation}\label{eq:tau}
    \tau_\lambda=\frac{\mathcal{P}_\lambda(\dya{\rm hw})}{d_\lambda}\,.
\end{equation}
\end{proposition}

While most of our proofs are presented in the Appendix, we explicitly derive Proposition~\ref{prop:tau} (as well as other main results) in the main text as the derivation provides key insights. In particular, the following  derivation of Proposition~\ref{prop:tau} contains an explicit recipe for constructing specific SW kernels that will be used in the examples below.

\begin{proof}
  We begin by taking $\Omega=e$. Here we find 
\begin{equation}
    \bra{e}D^\lambda_j \ket{e}=\bra{\rm hw}D^\lambda_j \ket{{\rm hw}}=\begin{cases}\langle D^\lambda_j\rangle \,, \quad \text{if $D^\lambda_j\in V_\lambda^0$ $\forall \lambda$, }\\0\,, \quad \text{otherwise.} \end{cases}\nonumber
\end{equation}
with $V_\lambda^0$ the diagonal weight zero subspace of $V_\lambda$. Note that above, we have simply defined the shorthand notation $\langle D^\lambda_j\rangle=\bra{\rm hw}D^\lambda_j \ket{{\rm hw}}$ for the non-zero elements of the irrep. In particular, let us define as $J_\lambda^0$ the set of indexes such that $D_j^\lambda\in V_\lambda^0$, and we define as $d_\lambda^0=\dim(V_\lambda^0)=|J_\lambda^0|$ the dimension of the weight zero subspace of $V_\lambda$. As such
\begin{equation}\label{eq:decomp-at-e}
    \dya{e}=\dya{{\rm hw}}= \sum_\lambda\sum_{j\in J_\lambda^0}\langle D^\lambda_j\rangle  D_j^\lambda\,.
\end{equation}

Next, combining Eqs.~\eqref{eq:basistransf}  and~\eqref{eq:Omega-state} leads to
\begin{equation}\label{eq:covariant-action}
    T(\Omega)D_j^\lambda T\ad(\Omega)=\sum_{j'}\left(\Phi_{\rm Ad}^\lambda\left(\Omega\right)\right)_{j,j'} D_{j'}^\lambda\,,
\end{equation}
where $\Phi_{\rm Ad}^\lambda$ is the irreducible unitary representation of $\mathbb{G}$ defined above. Hence, Eq.~\eqref{eq:decomp-at-e} shows that 
\begin{align}
    \dya{\Omega}&= \sum_\lambda\sum_{j\in J_\lambda^0} \sum_{j'} \langle D^\lambda_j\rangle\left(\Phi_{\rm Ad}^\lambda\left(\Omega\right)\right)_{j,j'}  D_{j'}^\lambda \,.
\end{align}
From the defining relation~\eqref{eq:tau0} we identify 
\begin{equation}\label{eq:hamonics}
    Y_{j'}^\lambda(\Omega)=\tau_\lambda^{-1/2}\sum_{j\in J_\lambda^0}\langle D^\lambda_j\rangle \left(\Phi_{\rm Ad}^\lambda\left(\Omega\right)\right)_{j,j'}\,.
\end{equation}
Using the orthogonality of the harmonic functions of Eq.~\eqref{eq:ortogonal-Y}, we find
\small
\begin{align}
   \delta_{l'j'} 
  &=\left\langle Y_{l'}^{\lambda},Y_{j'}^{\lambda}\right\rangle_{L^2}\nonumber\\
  &=\int_{\XC} d\mu(\Omega)\tau_\lambda^{-1}\!\!\sum_{l,j\in J_\lambda^0}\!\!\langle D_\lambda^l\rangle\langle D_{\lambda}^j\rangle \left(\Phi_{\rm Ad}^\lambda\!\left(\Omega\right)\right)_{l,l'}\left(\Phi_{\rm Ad}^\lambda\left(\Omega\right)\right)_{j,j'}\nonumber\\
  &=\tau_\lambda^{-1}\sum_{l,j\in J_\lambda^0}\langle D_\lambda^l\rangle\langle D_{\lambda}^j\rangle \frac{\delta_{l,j}\delta_{l',j'}}{d_\lambda}\nonumber\\
  &=\tau_\lambda^{-1}\frac{1}{d_\lambda}\sum_{j\in J_\lambda^0}\langle D^\lambda_j\rangle^2 \delta_{l',j'}\,,
\end{align}
\normalsize
where we used the orthogonality of the elements of irreducible representations. The previous allows us to provide the general formula
\begin{equation}\label{eq:tau-prood}
    \tau_\lambda=\frac{1}{d_\lambda}\sum_{j\in J_\lambda^0}\langle D^\lambda_j\rangle^2=\frac{\mathcal{P}_\lambda(\dya{\rm hw})}{d_\lambda}\,.
\end{equation}
For the special case when the expectation values are independent of $l$, i.e., $\langle D_\lambda^l\rangle=\langle D_\lambda\rangle$ we obtain  
\begin{equation}\label{eq:tauweight0dims}
    \tau_\lambda
    = \frac{d_\lambda^0}{d_\lambda}\,\langle D_\lambda\rangle^2\,.
\end{equation}
\end{proof}

Equation~\eqref{eq:tau} in Proposition~\ref{prop:tau} shows that $\tau_\lambda$ is simply the GFD purity of the
reference free state $\dya{{\rm hw}}$ in the irrep labeled by $\lambda$, divided by the
dimension $d_\lambda$ of that irrep. Since $\sum_\lambda \mathcal{P}_\lambda(\dya{{\rm hw}})=1$
for a pure state, the numbers $\{d_\lambda\tau_\lambda\}_\lambda$ form a probability
distribution over irreps, and $d_\lambda\tau_\lambda$ is exactly the fraction of the
operator-space purity of the density matrix $\dya{{\rm hw}}$ carried by the sector $V_\lambda$.
Equivalently, $\tau_\lambda$ quantifies how strongly the highest-weight state populates the irrep
 $\lambda$ in $\LC(\HC)$.
To finish, we note that we can combine Eqs.~\eqref{eq:SW-kernel} and~\eqref{eq:tau-prood}
\small
\begin{align}
\Delta(\Omega,s)&=\sum_\lambda\sum_{j=1}^{d_\lambda}\tau_\lambda^{-(s+1)/2}\left(\sum_{j'\in J_\lambda^0}\langle D^\lambda_{j'}\rangle \left(\Phi_{\rm Ad}^\lambda\left(\Omega\right)\right)_{j',j}\right)D_j^\lambda\nonumber\\
&=\sum_\lambda\tau_\lambda^{-(s+1)/2}\sum_{j'\in J_\lambda^0}\langle D^\lambda_{j'}\rangle\left(\sum_{j=1}^{d_\lambda} \left(\Phi_{\rm Ad}^\lambda\left(\Omega\right)\right)_{j',j}D_j^\lambda\right)\nonumber\\
&=T(\Omega)\left(\sum_\lambda\tau_\lambda^{-(s+1)/2}\sum_{j'\in J_\lambda^0}\langle D^\lambda_{j'}\rangle D_{j'}^\lambda\right) T\ad(\Omega)\,.\label{eq:delta-explicit}
\end{align}
\normalsize
In the last line, we used the covariance of the group action as per Eq.~\eqref{eq:covariant-action}. Indeed, setting $s=-1$ in Eq.~\eqref{eq:delta-explicit} readily recovers 
\begin{align}
\Delta(\Omega,-1)&=T(\Omega)\left(\sum_\lambda\sum_{j'\in J_\lambda^0}\langle D^\lambda_{j'}\rangle D_{j'}^\lambda\right) T\ad(\Omega)\nonumber\\
&=T(\Omega)\dya{{\rm hw}}T\ad(\Omega)=\dya{\Omega}\,.\nonumber
\end{align}

\subsection{Irrep decomposition of the SW kernel}

In this brief section we note that since the SW kernel is an operator in $\LC(\HC)$, we can compute its GFD purities according to Definition~\ref{def:GFD-purities-LH}. In particular, a straightforward calculation (see the Appendix for a proof) leads to the following result.
\begin{proposition}\label{prop:irreps-Delta-LH}
    Let $\Delta(\Omega,s)$ be an SW kernel as defined in Eq.~\eqref{eq:SW-kernel}, then its GFD purities in the irreps of $\LC(\HC)$ are
\begin{equation}\label{eq:purity-SW-kernel}
    \PC_\lambda(\Delta(\Omega,s))=\tau_\lambda^{-s} d_\lambda=\left(\frac{\mathcal{P}_\lambda(\dya{\rm hw})}{d_\lambda}\right)^{-s}d_\lambda\,.
\end{equation}
\end{proposition} 

At $s=0$ the kernel purity is proportional to $d_\lambda$, while at $s=-1$ it equals the reference-state purity. For $s=1$, sectors with smaller $\tau_\lambda$ receive greater weight. Since all retained $\tau_\lambda$ are positive, none of these sectors is removed by the SW correspondence map.

\subsection{QPS kernels as signal-processing filters}

With the previous result in mind, we are now ready to study the irrep decomposition of $F_{\rho}(\Omega,s)$. Combining Eqs.~\eqref{eq:PS-function} and~\eqref{eq:SW-kernel} we obtain
\begin{equation}
    F_{\rho}(\Omega,s)
    = \sum_\lambda \sum_{j=1}^{d_\lambda}
      \tau_\lambda^{-s/2}\,Y^\lambda_j(\Omega)\,
      \big\langle D_j^\lambda,\rho\big\rangle_{\LC}\,.
\end{equation}
Then one can prove the following proposition (see the Appendix for a derivation).
\begin{proposition}\label{prop:irreps-rho-LX}
    Let $\rho$ be a quantum state and $F_{\rho}(\Omega,s)$ its phase space representation. Then, from Definition~\ref{def:GFD-purities-LX}, the irrep projections are non-zero only for the irreps $\sigma=\lambda$ that appear in the decomposition of $\LC(\HC)$ in Definition~\ref{def:GFD-purities-LH}. In this case we have
    \begin{equation}
        [F_{\rho}]_\lambda(\Omega,s)
        = \sum_{j=1}^{d_\lambda}
          \tau_\lambda^{-s/2}\,Y^\lambda_j(\Omega)\,
          \big\langle D_j^\lambda,\rho\big\rangle_{\LC}\,,
    \end{equation}
    and the ensuing purities
    \begin{align}
        \widetilde{\PC}_\lambda\big(F_{\rho}(\Omega,s)\big)
        &= \tau_\lambda^{-s}\,\PC_\lambda(\rho)\nonumber\\
        &= \left(\frac{\mathcal{P}_\lambda(\dya{{\rm hw}})}{d_\lambda}\right)^{-s}
           \PC_\lambda(\rho)\,.
    \end{align}
\end{proposition}

Proposition~\ref{prop:irreps-rho-LX} connects the GFD profile of a state in operator space to its profile in each phase-space representation. Combined with Proposition~\ref{prop:tau}, it gives a general construction of the corresponding filters directly from the reference free-state purities. Once that profile is known, the construction fixes the multipliers applied to every state's GFD profile throughout the SW family. At $s=0$ the profile is unchanged. At $s=-1$ sectors with smaller $\tau_\lambda$ are attenuated more strongly, and at $s=1$ they are amplified more strongly.

The terms low-pass and high-pass refer to an ordering by harmonic frequency, naturally specified by Casimir or Laplace--Beltrami eigenvalues. Importantly, we note that irrep dimension does not in general supply a unique such ordering, and Proposition~\ref{prop:tau} alone does not establish monotonicity of $\tau_\lambda$. In all examples considered here, however, $\tau_\lambda$ decreases along the relevant frequency ordering, so this terminology is justified. Moreover we note that the filtering changes the representation of a fixed state, not its resourcefulness. The finite-dimensional maps considered are injective, so attenuation does not mean loss of information.

\subsection{Duality between free states and Haar-random typical states}

Here we consider Haar-random pure states as a reference ensemble of typical states. In most of the QRTs considered, these states are also highly resourceful.

For the reference free state, Proposition~\ref{prop:irreps-rho-LX} gives
\begin{align}
&\widetilde{\PC}_\lambda(F_{\dya{\rm hw}}(\cdot,-1))
=\frac{\PC_\lambda(\dya{\rm hw})^2}{d_\lambda},\\
&\widetilde{\PC}_\lambda(F_{\dya{\rm hw}}(\cdot,1))=d_\lambda.
\end{align}
Thus the Husimi representation weights its already populated sectors more strongly relative to the others, whereas the $s=1$ profile is proportional to irrep dimension.

Next, consider a Haar random state $\ket{\psi_H}$, and let us note that we can always express the GFD purity as
\begin{equation}
    \PC_\lambda(\dya{\psi_H})=\sum_{j=1}^{d_\lambda} \Tr[(D_j^{\lambda})^{\otimes 2}\dya{\psi_H}^{\otimes 2}]\,.
\end{equation}
From here we can find that, in expectation value~\cite{mele2023introduction,garcia2023deep}, 
\begin{equation}
    \mathbb{E}_{\HC}[\dya{\psi_H}^{\otimes 2}]=\frac{I^{\otimes 2}+{\rm SWAP}}{d(d+1)}\,,
\end{equation}
with $I$ the identity operator in $\LC(\HC)$ and ${\rm SWAP}$ the swap operator acting on $\HC^{\otimes 2}$. Hence,  if $\lambda$ is the trivial irrep containing $I$ we obtain $\mathbb{E}_{\HC}[\PC_\lambda(\dya{\psi_H})]=1/d$, while for every other irrep
\begin{equation}\label{eq:purity-Haar}
    \mathbb{E}_{\HC}[\PC_\lambda(\dya{\psi_H})]=\frac{d_\lambda}{d(d+1)}\,.
\end{equation}
Thus the mean purity per basis element is the same in every nontrivial sector. The total sector purity is proportional to $d_\lambda$. 

For every nontrivial irrep $\lambda\neq0$, Proposition~\ref{prop:irreps-rho-LX} gives 
   \begin{align}\label{eq:connections-1}
        \mathbb{E}_{\HC}[\widetilde{\PC}_\lambda\big(F_{\dya{\psi_H}}(\Omega,-1)\big)
        ]&= \frac{\widetilde{\PC}_\lambda\big(F_{\dya{{\rm hw}}}(\Omega,0)\big)}{d(d+1)}\,.
    \end{align}
Moreover, since 
\begin{equation}
   \mathbb{E}_{\HC}[\widetilde{\PC}_\lambda\big(F_{\dya{\psi_H}}(\Omega,0)\big)
        ]= \mathbb{E}_{\HC}[\PC_\lambda(\dya{\psi_H})]\,,\nonumber
\end{equation}
we also find 
\begin{equation}\label{eq:connections-2}
   \mathbb{E}_{\HC}[\widetilde{\PC}_\lambda\big(F_{\dya{\psi_H}}(\Omega,0)\big)]=\frac{\widetilde{\PC}_\lambda\big(F_{\dya{{\rm hw}}}(\Omega,1)\big)}{d(d+1)}\,.
\end{equation}
These are two instances of the following identity, under the assumptions of the scalar SW construction.

\begin{proposition}\label{prop:duality}
Let $\ket{\psi_H}$ be a random pure state drawn from the unitarily invariant probability measure induced by Haar sampling on $U(d)$, and let $\ket{{\rm hw}}$ be the highest weight free state of the QRT. For every nontrivial irrep $\lambda\neq0$ and every real $s$, their QPS purities satisfy the ensemble-average relation    \begin{equation}\label{eq:connections-general}
    \mathbb{E}_{\HC}\big[\widetilde{\PC}_\lambda\big(F_{\dya{\psi_H}}(\Omega,s)\big)\big]
    = \frac{\widetilde{\PC}_\lambda\big(F_{\dya{{\rm hw}}}(\Omega,s+1)\big)}{d(d+1)}\,.
\end{equation}
\end{proposition}

\begin{proof}
Using Proposition~\ref{prop:irreps-rho-LX}
and Eq.~\eqref{eq:purity-Haar}, which applies only to $\lambda\neq0$, we have for arbitrary $s$
\begin{equation}
    \mathbb{E}_{\HC}\big[\widetilde{\PC}_\lambda\big(F_{\dya{\psi_H}}(\Omega,s)\big)\big]
    = \frac{d_\lambda\,\tau_\lambda^{-s}}{d(d+1)}\,.\nonumber
\end{equation}
On the other hand, for the highest-weight state $\dya{{\rm hw}}$ we know that
$\PC_\lambda(\dya{{\rm hw}})=d_\lambda \tau_\lambda$, so
\begin{equation}
    \widetilde{\PC}_\lambda\big(F_{\dya{{\rm hw}}}(\Omega,s+1)\big)
    = \tau_\lambda^{-(s+1)}\,\PC_\lambda(\dya{{\rm hw}}) = d_\lambda\,\tau_\lambda^{-s}.\nonumber
\end{equation}
Combining these two expressions yields the compact relation
\begin{equation}
    \mathbb{E}_{\HC}\big[\widetilde{\PC}_\lambda\big(F_{\dya{\psi_H}}(\Omega,s)\big)\big]
    = \frac{\widetilde{\PC}_\lambda\big(F_{\dya{{\rm hw}}}(\Omega,s+1)\big)}{d(d+1)}\,.\nonumber
\end{equation}
Equations~\eqref{eq:connections-1} and~\eqref{eq:connections-2} correspond to the
special cases $s=-1$ and $s=0$ of~\eqref{eq:connections-general}. 
\end{proof}

The trivial sector must be treated separately. With $D^0=I/\sqrt d$ and the probability measure used here, $\tau_0=1/d$ and $Y^0=1$. Every normalized state therefore satisfies
\begin{equation}\label{eq:identity-sector}
\PC_0(\rho)=\frac1d,\qquad
\widetilde{\PC}_0(F_\rho(\cdot,s))=d^{s-1}.
\end{equation}
Here one can see that the trivial irrep does not obey Eq.~\eqref{eq:connections-general}, meaning that the duality is a relation between the nonconstant parts of the SW correspondences.

\subsection{Additional implications for QPSs}
We collect several direct consequences of the multiplier relation. They provide useful interpretations of standard SW correspondence-calculus results.

\subsubsection{Characterization of the Cahill-Glauber parameter $s$-flow}
Here we discuss how the purities $\widetilde{\PC}_\lambda(F_\rho(\Omega,s))$ are expected to change as a function of the continuous Cahill-Glauber  parameter $s$. Starting from
\begin{align}
\widetilde{\PC}_\lambda\big(F_{\rho}(\Omega,s)\big)= \tau_\lambda^{-s}\,\PC_\lambda(\rho)\,,
\end{align}
and taking the derivative with respect to $s$ leads to
\begin{equation}
    \frac{d}{ds}\widetilde{\PC}_\lambda\big(F_{\rho}(\Omega,s)\big)=-\log(\tau_\lambda) \widetilde{\PC}_\lambda\big(F_{\rho}(\Omega,s)\big)\,.
\end{equation}
Thus each sector follows an exponential $s$-flow with rate $-\log\tau_\lambda$. This is a direct differentiation of the multiplier relation, with the rates fixed by the reference-state purities.

\subsubsection{Twisted product and algebra structure in QPS}
The SW correspondence also transports operator multiplication to its finite-dimensional image in $L^2(\XC)$. The resulting twisted product is the standard noncommutative product~\cite{brif1999phase},

\begin{align}
  F_{AB}(\Omega,s_1)
  &= \big(F_A \star_{s_1,s_2,s_3} F_B\big)(\Omega)\nonumber\\
  &= \int_{\XC}\!\!d\mu(\Omega')\!\int_{\XC}\!\!d\mu(\Omega'')
     M_{s_1,s_2,s_3}(\Omega,\Omega',\Omega'')\nonumber\\
     &\quad\quad\quad\quad\quad\quad\quad \times F_A(\Omega',s_2)\,F_B(\Omega'',s_3),\nonumber
\end{align}
where the twisted-product kernel
\small
\begin{equation}
 M_{s_1,s_2,s_3}(\Omega,\Omega',\Omega'')
  \!=\! \Tr\!\big[\Delta(\Omega,s_1)\,\Delta(\Omega',-s_2)\,\Delta(\Omega'',-s_3)\big]\nonumber
\end{equation}
\normalsize
is completely determined by the SW kernels. In this way, the QPS representation
identifies $\LC(\HC)$ with its SW correspondence image as a $\mathbb{G}$-covariant algebra~\cite{brif1999phase}.

Expanding each kernel in the harmonic basis as in Eq.~\eqref{eq:SW-kernel} allows us to write the twisted-product kernel as

    \begin{align}
  M_{s_1,s_2,s_3}(\Omega,\Omega',\Omega'')
  \!&= \!\!\sum_{\lambda_1,\lambda_2,\lambda_3}
    T_{s_1,s_2,s_3}^{\lambda_1\lambda_2\lambda_3}
    \!\!\sum_{j_1,j_2,j_3}\Big(
      C^{\lambda_1\lambda_2\lambda_3}_{j_1 j_2 j_3}\nonumber\\
    & \times 
      Y_{j_1}^{\lambda_1}(\Omega)\,
      Y_{j_2}^{\lambda_2}(\Omega')\,
      Y_{j_3}^{\lambda_3}(\Omega'')\big),
\end{align}
where
\begin{align}
  C^{\lambda_1\lambda_2\lambda_3}_{j_1 j_2 j_3}
  =& \Tr\big[D_{j_1}^{\lambda_1} D_{j_2}^{\lambda_2} D_{j_3}^{\lambda_3}\big]\,,\\
  T_{s_1,s_2,s_3}^{\lambda_1\lambda_2\lambda_3}
  =& \tau_{\lambda_1}^{-s_1/2}\,\tau_{\lambda_2}^{s_2/2}\,\tau_{\lambda_3}^{s_3/2}\\
   =& \left(\frac{\mathcal P_{\lambda_1}(\dya{{\rm hw}})}{d_{\lambda_1}}\right)^{-s_1/2}\nonumber\\
   &\times\left(\frac{\mathcal P_{\lambda_2}(\dya{{\rm hw}})}{d_{\lambda_2}}\right)^{s_2/2}\nonumber\\
   &\times
     \left(\frac{\mathcal P_{\lambda_3}(\dya{{\rm hw}})}{d_{\lambda_3}}\right)^{s_3/2}\,.
\end{align}

Importantly, the coefficients $C^{\lambda_1\lambda_2\lambda_3}_{j_1 j_2 j_3}$ are purely group-theoretic and arise from the bases of the irreps, whereas the terms $T_{s_1,s_2,s_3}^{\lambda_1\lambda_2\lambda_3}$ are coupling constants which depend only on the GFD purity profile in $\LC(\HC)$ of the highest-weight state.

\subsection{Norm bounds for QPS representations from highest-weight purities}

The relation
\begin{equation}
    \widetilde{\mathcal P}_\lambda\big(F_\rho(\Omega,s)\big)
    = \tau_\lambda^{-s}\,\mathcal P_\lambda(\rho),
\end{equation}
together with the decomposition of the $L^2(\XC)$ norm into GFD purities (via the Plancherel Theorem~\cite{folland2016course}),
\begin{equation}
    \|F_\rho(\cdot,s)\|_{L^2}^2=\langle F_\rho(\cdot,s),F_\rho(\cdot,s)\rangle_{L^2}
    = \sum_\lambda \widetilde{\mathcal P}_\lambda\big(F_\rho(\Omega,s)\big),\nonumber
\end{equation}
gives elementary bounds by the smallest and largest filter weights.

For a general state $\rho$ we have
\begin{equation}
     \|F_\rho(\cdot,s)\|_{L^2}^2
    = \sum_\lambda \tau_\lambda^{-s}\,\mathcal P_\lambda(\rho),
\end{equation}
which, using $0 \le \mathcal P_\lambda(\rho)\le\Tr[\rho^2]$ and
$\sum_\lambda\mathcal P_\lambda(\rho)=\Tr[\rho^2]$, implies
\begin{equation}
    \min_\lambda \tau_\lambda^{-s}\,\Tr[\rho^2]
    \;\le\;
    \|F_\rho(\cdot,s)\|_{L^2}^2
    \;\le\;
    \max_\lambda \tau_\lambda^{-s}\,\Tr[\rho^2].\nonumber
\end{equation}
In particular, for a pure state $\rho=\dyad{\psi}$, where 
$\Tr[\rho^2]=1$ and $\sum_\lambda \mathcal P_\lambda(\dyad{\psi})=1$, we
obtain the pure-state specialization
\begin{equation}
    \min_\lambda \tau_\lambda^{-s}
    \;\le\;
    \|F_{\dyad{\psi}}(\cdot,s)\|_{L^2}^2
    \;\le\;
    \max_\lambda \tau_\lambda^{-s}.
\end{equation}
The extrema run over the retained sectors, all of which have $\tau_\lambda>0$. Importantly, these inequalities are consequences of bounding a weighted sum by its extreme weights and they need not be tight for physical states.

\subsection{Additional connections to QRTs}
In the previous section we have discussed how the choice for the Cahill-Glauber  parameter $s$ serves as an effective filter in the Fourier decomposition of a quantum state $\rho$. Here we provide further interpretation to this phenomenon by realizing that when $s=-1$, then 
\begin{equation}
F_\rho(\Omega,-1)=\bra{\Omega}\rho\ket{\Omega}\,,
\end{equation}
which is nothing more than the fidelity between $\rho$ and the coherent states. That is, from a QRT perspective $F_\rho(\Omega,-1)$  can be understood as a comparison between $\rho$ and the orbit of free states with a low-pass filter.  Indeed, one can then re-write the coherent state fidelity of Eq.~\eqref{eq:fidelity} as
\begin{align}\label{eq:fidelity-2}
S(\ket{\psi})&=\max_{\Omega\in\XC}F_\rho(\Omega,-1)\,.
\end{align}

For the examples below, higher Casimir eigenvalues describe more rapidly varying harmonics on the coherent-state orbit. The Husimi coefficients attenuate these harmonics, while Eq.~\eqref{eq:kernel-Deltas} relates the remaining representations to the same fidelity function:
\begin{equation}
F_\rho(\Omega,s)=\int_{\XC}d\mu(\Omega')K_{s,-1}(\Omega,\Omega')\bra{\Omega'}\rho\ket{\Omega'}.
\end{equation}
This gives a mathematical connection between the choice of SW correspondence and its comparison with the orbit of free states.

\section{Examples}

Here we showcase how QPSs serve as filters for the Fourier components of states across different QRTs. 

\subsection{Spin coherence}

We begin by considering the spin coherence QRT arising in quantum systems with total angular momentum $S$, whose dynamics are governed by the irreducible representation of $\SU(2)$. Notably, such states can be experimentally prepared (e.g., in nuclear magnetic resonance systems)~\cite{nielsen2000quantum,arecchi1972atomic}, and serve as a basis for macroscopic quantum information protocols~\cite{byrnes2015macroscopic,pyrkov2014quantum,nielsen2000quantum}. Not only is this example relevant, simple and illustrative, but it also allows for a direct phase space visualization. Let $\HC=\mathbb{C}^d={\rm span}_{\mathbb{C}}\{\ket{S,m}\}_{m=-S}^S$ be the spin $S=(d-1)/2$ irreducible representation of $\mathbb{G}=\SU(2)$, with associated Lie algebra $\mathfrak{g}=\mathfrak{su}(2)$. That is, $T(\mathfrak{g})={\rm span}_{\mathbb{R}}\{J_x,J_y,J_z\}$ with
\begin{align}
    J_z&={\rm diag}(S,S-1,\ldots,-S+1,-S)\,,\nonumber\\
    J_x&=\frac{J_++J_-}{2}\,,\quad J_y=\frac{J_+-J_-}{2i}    
\end{align}
and $J_\pm$ the standard spin raising and lowering operators. The highest weight of the algebra (associated to the Cartan $J_z)$ is the state $\ket{\rm hw}=\ket{S,S}$. Then, the space of operators $\LC(\HC)$ decomposes into $2S+1$ multiplicity-free irreps as
\begin{align}
    \LC(\HC)=\bigoplus_{\lambda=0}^{2S} V_\lambda\,,
\end{align}
where $V_\lambda$ is a spin-$\lambda$ irrep of dimension $\dim(V_\lambda)=2\lambda+1$, with the standard complex orthonormal tensor basis
\begin{equation}\label{eq:su2_irrep_basis}
D^\lambda_{j}= \sum_{j' =-S}^S (-1)^{S-(j'-j)} c^{S,S,\lambda}_{j',j-j',j} \ketbra{S,j'}{S,j'-j} \,.
\end{equation}
Above, $c_{j_1,j_2,j}^{s_1,s_2,s}$ are Clebsch-Gordan coefficients and $j=-\lambda,\ldots,\lambda$. These tensors obey $(D_j^\lambda)^\dagger=(-1)^jD_{-j}^\lambda$; real and imaginary combinations give the Hermitian basis used in the general derivations. Indeed, sector purities are unchanged by this basis change. The harmonic frequency is ordered by the spherical Laplacian eigenvalue $\lambda(\lambda+1)$, and $\tau_\lambda$ in Eq.~\eqref{eq:tau-su2} decreases with $\lambda$. The weight zero subspace $V_\lambda^0\subseteq V_\lambda$ is of dimension one, i.e., $d_\lambda^0=1$, as  $V_\lambda^0={\rm span}_{\mathbb{C}}\{D^\lambda_{0}\}$, and we find
\begin{equation}
    \bra{S,S}D^\lambda_{0}\ket{S,S}=c^{S,S,\lambda}_{S,-S,0}\,,
\end{equation}
from where, using Eq.\ \eqref{eq:tauweight0dims}, one finds 
\begin{equation}\label{eq:tau-su2}   \tau_\lambda=\frac{\left(c^{S,S,\lambda}_{S,-S,0}\right)^2}{2\lambda +1}\,.
\end{equation}

\begin{figure*}[t]
    \centering
\includegraphics[width=.97\linewidth]{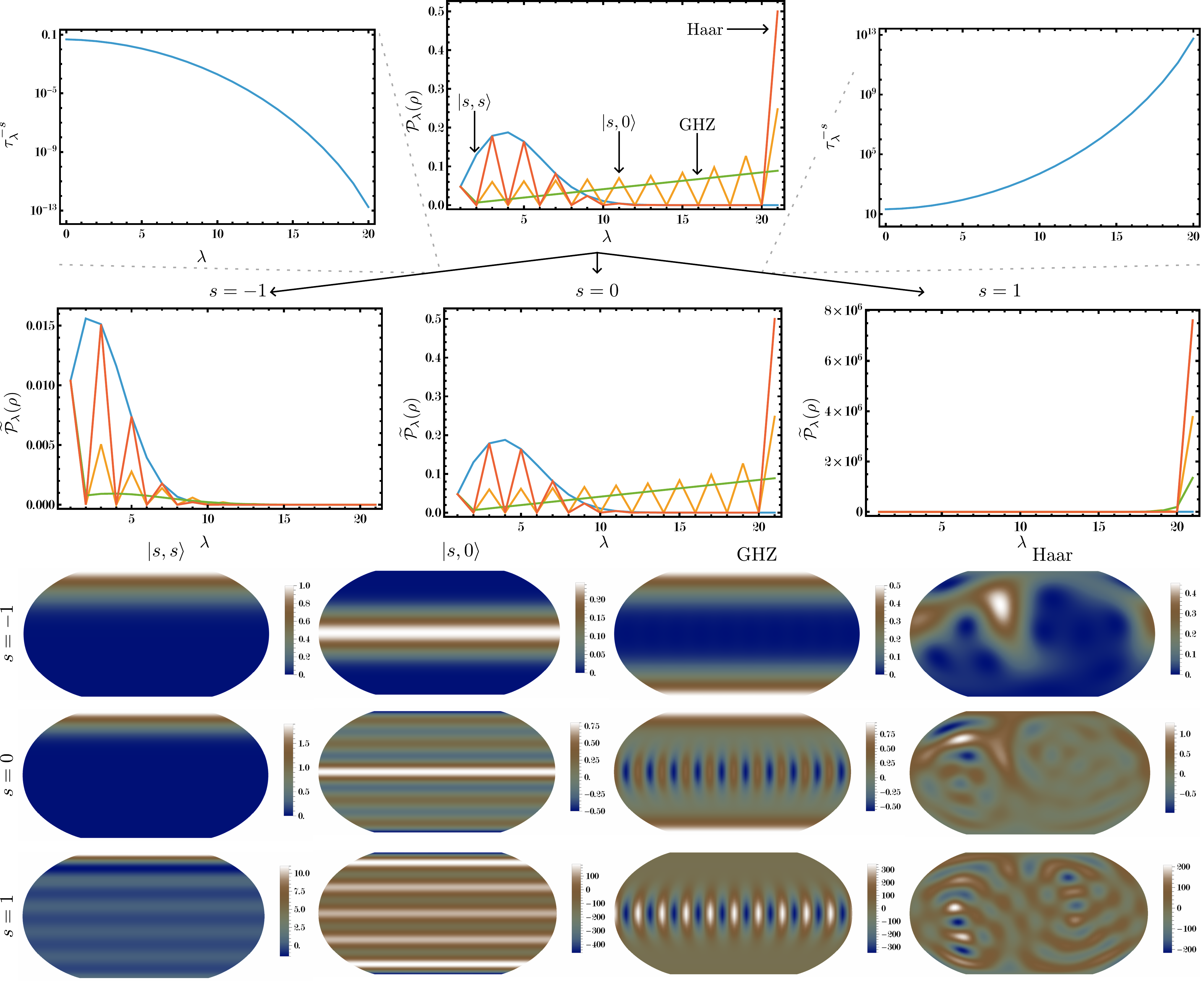}
    \caption{\textbf{SU(2) spin-coherence QPS for $S=10$.}
(a, top) The power multipliers $\tau_\lambda$ and $\tau_\lambda^{-1}$ and the operator-space purities for $\ket{S,S}$, $\ket{S,0}$, GHZ, and the Haar ensemble mean. The labels are the actual spins $\lambda=0,\ldots,10$.
(a, middle) Phase-space purities $\widetilde\PC_\lambda=\tau_\lambda^{-s}\PC_\lambda$ at $s=-1,0,1$.
(b) Robinson projections of the corresponding SW correspondence, using one fixed Haar-random pure state for the last column. All panels use the probability-measure normalization of Section~\ref{sec:normalization}. The plots have been recomputed with this convention and the corrected GHZ formula.}
    \label{fig:fig-su(2)}
\end{figure*}

Then, we recall that the QPS is obtained by taking the quotient of $\mathbb{G}$ and the subgroup $\mathbb{K}$  which stabilizes $\ket{{\rm hw}}$. One can readily verify that $\mathbb{K}=U(1)$,  representation $e^{i\phi J_z}$, meaning that
\begin{equation}
    \XC=SU(2)/U(1)\simeq S^2\,.
\end{equation}
and indicating that the QPS is isomorphic to a sphere. This allows us readily parametrize $\XC$ by a polar and azimutal angles, i.e. $\Omega=(\theta,\phi)$ such that
\begin{equation}
    \ket{\Omega}=e^{-i\phi J_z}e^{-i\theta J_y}\ket{S,S}
\end{equation}
and the Haar measure over $\XC$ is then given
\begin{equation}
   \int_{\XC} d\mu(\Omega)\equiv\frac{1}{4\pi}\int_{0}^{2\pi} d\phi\int_{0}^{\pi} d\theta\sin(\theta)\,.
\end{equation}
Putting it all together, we find from Eq.~\eqref{eq:delta-explicit}
\begin{align}
\Delta((\theta,\phi),s)=e^{-i\phi J_z}e^{-i\theta J_y}\Delta((0,0),s)e^{i\theta J_y} e^{i\phi J_z}\,,
\end{align}
with
\begin{equation}
   \Delta((0,0),s)=\sum_{\lambda=0}^{2S}\frac{(2\lambda +1)^{(s+1)/2}}{\left(c^{S,S,\lambda}_{S,-S,0}\right)^{s}}D_{0}^\lambda\,.
\end{equation}

With the previous results in hand, we are now ready to understand how different choices of SW kernels lead to different representations of quantum states in the QPSs. In particular, we will consider the highest weight resource-free state $\ket{S,S}$, as well as the resourceful states $\ket{S,0}$, $|GHZ\rangle =\frac{\ket{S,S}+\ket{S,-S}}{\sqrt{2}}$ and a Haar random state $\ket{\psi_H}$. Crucially, beyond the highest weight $\ket{S,S}$, which is considered a free spin coherent state, all others possess resource in the associated QRT. As shown in~\cite{bermejo2025characterizing}, their GFD purities in $\LC(\HC)$ are as follows: Taking $\rho_{m}=\dya{S,m}$,  one obtains
\begin{align}\label{eq:pur-sm}
    \PC_{\lambda}(\rho_{m})=\left( c_{m,-m,0}^{S,S,\lambda}\right)^2\,.
\end{align}
For the GHZ state $\rho_{\rm GHZ}=\dya{\rm GHZ}$, the diagonal terms contribute only at even $\lambda$, while its two off-diagonal coherences contribute at $\lambda=2S$. Thus
\begin{equation}\label{eq:pur-ghz-su2}
\PC_\lambda(\rho_{\rm GHZ})=
\begin{cases}
\left|c_{S,-S,0}^{S,S,\lambda}\right|^2,&\lambda\ \text{even},\\
0,&\lambda\ \text{odd}
\end{cases}
+\frac12\delta_{\lambda,2S}.
\end{equation}
Here $S$ is the total spin; it is distinct from the Cahill-Glauber parameter $s$.

Finally, for a Haar random state $\rho_H=\dya{\psi_H}$,  
\begin{align}\label{eq:su2-haar}
    \mathbb{E}_{\HC}[\mathcal{P}_{\lambda}(\rho_H)] = 
    \begin{cases}
        \dfrac{1}{d}\, & \text{if } \lambda = 0\,, \\
        \dfrac{d_\lambda}{d(d+1)}\, & \text{if } \lambda > 0\,,
    \end{cases}
\end{align}
as per Eq.~\eqref{eq:purity-Haar}. 

In Fig.~\ref{fig:fig-su(2)} we showcase different QPSs ($s=-1,0,1$) for the spin coherent QRT with $S=5$ for the aforementioned benchmark
states: $\ket{S,S}$, $\ket{S,0}$, $\ket{\mathrm{GHZ}}$ and $\ket{\psi_H}$.  First, in 
Panel~\ref{fig:fig-su(2)}(a, top center) we show the operator-space GFD purities
$P_\lambda(\rho)$ in $\LC(\HC)$ for each state. As anticipated by our previous discussions, we can see that the purities in the highest weight state (of Eq.~\eqref{eq:pur-sm})  are mostly concentrated in the lowest spins
$\lambda$. On the other hand, the resourceful $\ket{S,0}$ and the GHZ state populate progressively
higher irreps. Finally, the Haar ensemble mean grows as $d_\lambda$ in the nontrivial sectors.
The outer panels in the top row display the factors $\tau_\lambda$ and
$\tau_\lambda^{-1}$, computed from Eq.~\eqref{eq:tau-su2}, which quantify
how strongly each spin-$\lambda$ sector is attenuated or amplified when
moving to the $s=-1$ and $s=1$ QPs. As $\lambda$ increases, $\tau_\lambda$ decays
rapidly, while $\tau_\lambda^{-1}$ grows, making their role as low-pass and
high-pass filters explicit.

The middle row of Fig.~\ref{fig:fig-su(2)}(a) shows the corresponding
phase space purities $\widetilde P_\lambda(\rho,s)$. For $s=0$ the
$\widetilde P_\lambda(\rho,0)$ exactly coincide with the operator-space
purities, while for $s=-1$ the high-$\lambda$ components are attenuated more strongly than the low-$\lambda$ components: the kernel acts as a
genuine low-pass filter in the spin index. Conversely, for $s=1$ the
high-$\lambda$ sectors dominate, and the phase space representations of
states with substantial weight in the largest irreps (in particular GHZ and
the expected Haar purity) are almost entirely supported near the maximal spin
$\lambda=2S$, illustrating the high-pass character of the $s=1$ kernel.

Panel~\ref{fig:fig-su(2)}(b) displays the QPS functions
$F_\rho(\Omega,s)$ on $X\simeq S^2$ for the three
values $s=-1,0,1$ and for all considered states. We use a Robinson projection to map the sphere to a rectangle\footnote{Note that this projection is neither equal-area nor conformal, and the
distortion is severe near the poles but quickly decreases to moderate levels
away from them}. For the coherent free state $\ket{S,S}$, the $s=-1$ SW correspondence is a positive bump centered at the north pole. Increasing $s$ gives greater relative weight to higher spins and produces oscillations. In our normalization its $s=1$ value at the north pole is $\sum_{\lambda=0}^{2S}(2\lambda+1)=d^2=121$ for $S=5$. For $\ket{S,0}$ the SW correspondence display bands about the equator, with finer structure as the high-frequency components receive greater weight.

The GHZ state shows the strongest qualitative change with $s$. At $s=-1$ the
QPS is relatively smooth and concentrated in low-$\lambda$ sectors. At
$s=0$ a pronounced fringe pattern appears along the equator, and for $s=1$
the representation is almost entirely carried by the highest spin
$\lambda=2S$, giving rise to a highly oscillatory pattern that is essentially
the contribution of the $Y_{10,m}$ harmonics selected by the high-pass
filter. Finally, the Haar-random state $\ket{\psi_H}$ exhibits irregular,
speckle-like structures for all $s$, with increasing small-scale structure
as $s$ grows. The last column shows a single Haar draw, while the spectral curves in panel (a) show ensemble-averaged purities. Importantly, the two should not be identified: the mean SW correspondence itself is spatially constant. Haar sampling defines our typical-state ensemble, and each realization has its own harmonic pattern.

These examples also illustrate a direct connection between the irrep content
of a state and the ``smoothness'' of its QPS representation. In the
$\SU(2)$ case, the harmonics $Y_j^\lambda(\Omega)$ are spin-$\lambda$
spherical harmonics on $S^2$, so larger $\lambda$ correspond to eigenfunctions
of the Laplace--Beltrami operator with larger eigenvalues and hence to
more oscillatory, rapidly varying patterns on the sphere. Weight in
high-dimensional irreps therefore translates into increasingly complex
and fine-grained structure in $F_\rho(\Omega,s)$, while dominance of
low-$\lambda$ components yields smooth, slowly varying phase space
functions. This behavior is clearly visible in Fig.~\ref{fig:fig-su(2)}:
the coherent state $\ket{S,S}$, whose GFD purities are concentrated in
small $\lambda$, produces a smooth bump (especially at $s=-1$), whereas the
randomly chosen Haar state, which has significant support across many irreps
and is further biased towards high $\lambda$ at $s=1$, gives rise to highly
structured, rapidly fluctuating patterns on $S^2$.

To finish, we also note that (although not necessarily visible from Fig.~\ref{fig:fig-su(2)}), the
components in higher-dimensional irreps can be directly linked to the appearance of negativity in the
$s=0$ QPS representation. Indeed, this already occurs for free (spin-coherent) states. Since $\ket{{\rm hw}}$ has strictly non-zero components in  the  higher spins
$\lambda$, and since those  are not suppressed by the $s=0$ kernel
(i.e., $\tau_\lambda^0 = 1$), their oscillatory harmonics are fully visible
and can drive $F_\rho(\Omega,0)$ slightly negative, despite $\rho$ being
free in the spin-coherence QRT. By contrast, the $s=-1$ representation is
a genuine low-pass filter matched to the free sector and yields a manifestly
positive Q-function for all coherent states.

This offers a new perspective on results such as those in~\cite{davis2021wigner}, where Wigner negativity is reported even for spin-coherent states: within our framework, this is not a
paradoxical sign of ``resource'' in free states, but rather a consequence of
the undamped contribution of high-$\lambda$ modes in the $s=0$ kernel.
Negativity of the spin Wigner function can therefore occur through the interference of its harmonics even for free states; it is not by itself a faithful witness of resource in the spin-coherence QRT. Our group Fourier filter picture clarifies that the
choice of $s$ determines whether a given QPS representation behaves more
like a resource-sensitive representation ($s=1$) or a free-adapted,
nonnegative quasi-probability ($s=-1$).

For comparison with the standard continuous-variable setting, let $a$ be a single bosonic annihilation operator and use $D(\xi)=e^{\xi a^\dagger-\xi^*a}$. The Cahill-Glauber characteristic functions satisfy
\begin{equation}\label{eq:bosonic-filter}
\chi_s(\xi)=e^{s|\xi|^2/2}\chi_0(\xi),\qquad
\chi_0(\xi)=\Tr[\rho D(\xi)].
\end{equation}
Thus the Husimi representation damps Fourier amplitudes by $e^{-|\xi|^2/2}$, while the $P$ representation applies the inverse weight where the transform is well defined. In this convention the squared-amplitude coefficient is $\tau(\xi)=e^{-|\xi|^2}$~\cite{brif1999phase}.

The spin coefficients have the related limit
\begin{align}\label{eq:large-spin}
\frac{\tau_\lambda}{\tau_0}&=(2S+1)\tau_\lambda\nonumber\\
&=\prod_{k=0}^{\lambda-1}\frac{2S-k}{2S+k+2}\nonumber\\
&=\exp\!\left[-\frac{\lambda(\lambda+1)}{2S}+O(S^{-1})\right],
\end{align}
uniformly for $\lambda=O(\sqrt S)$. Identifying $|\xi|^2=\lambda(\lambda+1)/(2S)$ gives the same limiting relative coefficient $e^{-|\xi|^2}$. This comparison uses $\tau_\lambda/\tau_0$, since $\tau_0=1/(2S+1)$ vanishes in the large-spin limit. It explains the common low-pass and high-pass behavior. Importantly, we note that a full noncompact extension is outside our assumptions, as here  the plane has no invariant probability measure, and the bosonic $P$ function can be a distribution. As such, this example should simply be considered as a commmon point rather that a straightforward generalization.

\subsection{Other QRTs}
Given that we have already extensively showcased the effect of the QPS filters on the GFD purities for the spin coherence QRT, in this section we briefly present the basic ingredients needed to repeat the aforementioned analysis on two other widely used QRTs: multipartite entanglement and fermionic Gaussianity.

\subsubsection{Multipartite entanglement}

To begin, we consider the QRT of $n$-qubit multipartite entanglement~\cite{horodecki2009quantum,plenio2005introduction,beckey2021computable}, which plays a central role in quantum computation~\cite{ekert1998quantum,nielsen2000quantum,datta2005entanglement}, communication~\cite{bennett1993teleporting,barrett2002nonsequential,cleve1997substituting,gigena2017bipartite}, and sensing~\cite{paris2009quantum,huerta2022inference}. Let $\HC=(\mathbb{C}^2)^{\otimes n}$, and $\mathbb{G}=\mathbb{SU}(2)\times \mathbb{SU}(2)\times \cdots\times \mathbb{SU}(2)$, with $T$ the standard representation of each $\mathbb{SU}(2)$. That is, $T(g_1\times g_2\times \cdots \times g_n)=g_1\otimes g_2\otimes \cdots \otimes g_n$. The space of operators decomposes into $2^n$ multiplicity-free irreps as
\begin{equation}\label{eq:irreps-L-nq}
    \LC(\HC)=\bigoplus_{\lambda \in\{0,1\}^{\otimes n}} V_{\lambda}\,,
\end{equation}
where we can define an orthonormal basis of $V_{\lambda}$ as being composed of all operators taking the form $P_1^{\lambda_1}\otimes P_2^{\lambda_2}\otimes \cdots \otimes P_n^{\lambda_n}/\sqrt{2^n}$, for $P_i\in\{X,Y,Z\}$. From the previous, we see that $d_\lambda=3^{w(\lambda)}$, where $w(\lambda)$ is the Hamming weight of the bitstring $\lambda$. Then,  the weight-zero subspace $V_\lambda^0$ is composed of the operators $Z^\lambda=Z_1^{\lambda_1}\otimes Z_2^{\lambda_2}\otimes \cdots \otimes Z_n^{\lambda_n}/\sqrt{2^n}$, which means that $d_\lambda^0=1$. Since $\ket{{\rm hw}}=\ket{0}^{\otimes n}$ we find that $\forall \lambda\in\{0,1\}^{\otimes n}$
\begin{equation}
    \frac{1}{\sqrt{2^n}}\bra{0}^{\otimes n}Z_1^{\lambda_1}\otimes Z_2^{\lambda_2}\otimes \cdots \otimes Z_n^{\lambda_n} \ket{0}^{\otimes n}=\frac{1}{\sqrt{2^n}}\,,
\end{equation}
and hence
\begin{equation}
    \tau_\lambda=\frac{1}{3^{w(\lambda)} 2^n}\,.
\end{equation}
The product-sphere Laplacian has eigenvalue $2w(\lambda)$ on these local spin-$0/1$ sectors, so $\tau_\lambda$ decreases with frequency. In Fig.~\ref{fig:tau-others}(a) we plot $\tau_\lambda^{-s}$ versus the Hamming weight $w(\lambda)$ and we again see that for $s=-1$ the QPS serves as a low-pass filter as the higher-frequency sectors are attenuated. There we can also see that the converse occurs for $s=1$.

The QPS is again obtained as the homogeneous space $\XC = \mathbb{G}/\mathbb{K}$, where $\mathbb{K}$ is
the stabilizer of the highest-weight state. One can readily see that the stabilizer of each local $\ket{0}$ is the representation of $\U(1)$ given by $z$-rotations, so that
\begin{equation}
    \mathbb{K} = \U(1)\times\cdots\times\U(1) \simeq \U(1)^n,
\end{equation}
and hence~\cite{heightman2025quantum}
\begin{equation}
    \XC = \mathbb{G}/\mathbb{K} \simeq
    \big(\SU(2)/\U(1)\big)^{\times n}
    \simeq (S^2)^{\times n}.
\end{equation}
That is, the associated QPS is an $n$-fold product of two-dimensional spheres. We can therefore parametrize
\begin{equation}\label{eq:Omega-ent}
    \Omega = (\Omega_1,\dots,\Omega_n),\qquad
    \Omega_k = (\theta_k,\phi_k)\in S^2,
\end{equation}
and define the corresponding product coherent states as
\begin{equation}
    \ket{\Omega}
    = \bigotimes_{k=1}^n
      \ket{\Omega_k},\qquad
    \ket{\Omega_k}
    = e^{-i\phi_k Z_k/2}e^{-i\theta_k Y_k/2}\ket{0}_k,\nonumber
\end{equation}
with Haar measure
\begin{equation}
    \int_X d\mu(\Omega)
    = \frac{1}{(4\pi)^n}
      \prod_{k=1}^n
      \left(
          \int_0^{2\pi}\!d\phi_k
          \int_0^{\pi}\!d\theta_k\,\sin\theta_k
      \right).
\end{equation}
Hence, from Eq.~\eqref{eq:delta-explicit}
\begin{align}
\Delta(\Omega,s)&=T(\Omega)\!\left(\sum_{\lambda\in\{0,1\}^{\otimes n}}\!\!(3^{w(\lambda)})^{\frac{s+1}{2}}(2^n)^{\frac{s}{2}}Z^{\lambda}\right) \!T\ad(\Omega)\nonumber\\
&=T(\Omega)\left(\bigotimes_{k=1}^n2^{\frac{s-1}{2}}(I_k+3^{\frac{s+1}{2}}Z_k)\right) T\ad(\Omega)\,,
\end{align}
from where we can use Eq.~\eqref{eq:Omega-ent} and obtain
\begin{equation}
    \Delta(((\theta_1,\phi_1),\cdots,(\theta_n,\phi_n)),s)=\bigotimes_{k=1}^n\Delta((\theta_k,\phi_k),s)\,,
\end{equation}
with 
\begin{equation}\label{eq:delta-ent}
    \Delta((\theta_k,\phi_k),s)=R_z(\phi_k)R_y(\theta_k)\Delta((0,0),s)R_y\ad(\theta_k)R_z\ad(\phi_k)\,,\nonumber
\end{equation}
where
\begin{equation}
\Delta((0,0),s)=2^{\frac{s-1}{2}}\left(I+3^{\frac{s+1}{2}}Z\right)\,,    
\end{equation}
For one qubit, $\tau_0=1/2$ and $\tau_1=1/6$, and the displayed kernel has trace $2^{(s+1)/2}$. Its probability-measure integral is $2^{(s-1)/2}I$, exactly as required by our standardization. Dividing it by $2^{(s+1)/2}$ gives the customary trace-one kernel $[I+3^{(s+1)/2}Z]/2$.
Here we defined $R_z(\phi_k)=e^{-i\phi_k Z_k/2}$ and $R_y(\theta_k)=e^{-i\theta_k Y_k/2}$. We can verify that replacing $s=-1$ in the previous equation recovers 
\begin{align}
    \Delta((\theta_k,\phi_k),-1)&=\frac{1}{2}R_z(\phi_k)R_y(\theta_k)\left(I+Z\right)R_y\ad(\theta_k)R_z\ad(\phi_k)\nonumber\\
&=R_z(\phi_k)R_y(\theta_k)\dya{0}R_y\ad(\theta_k)R_z\ad(\phi_k)\nonumber\\
&=\dya{\Omega_k}\,.
\end{align}

\begin{figure}[t]
    \centering
    \includegraphics[width=.99\linewidth]{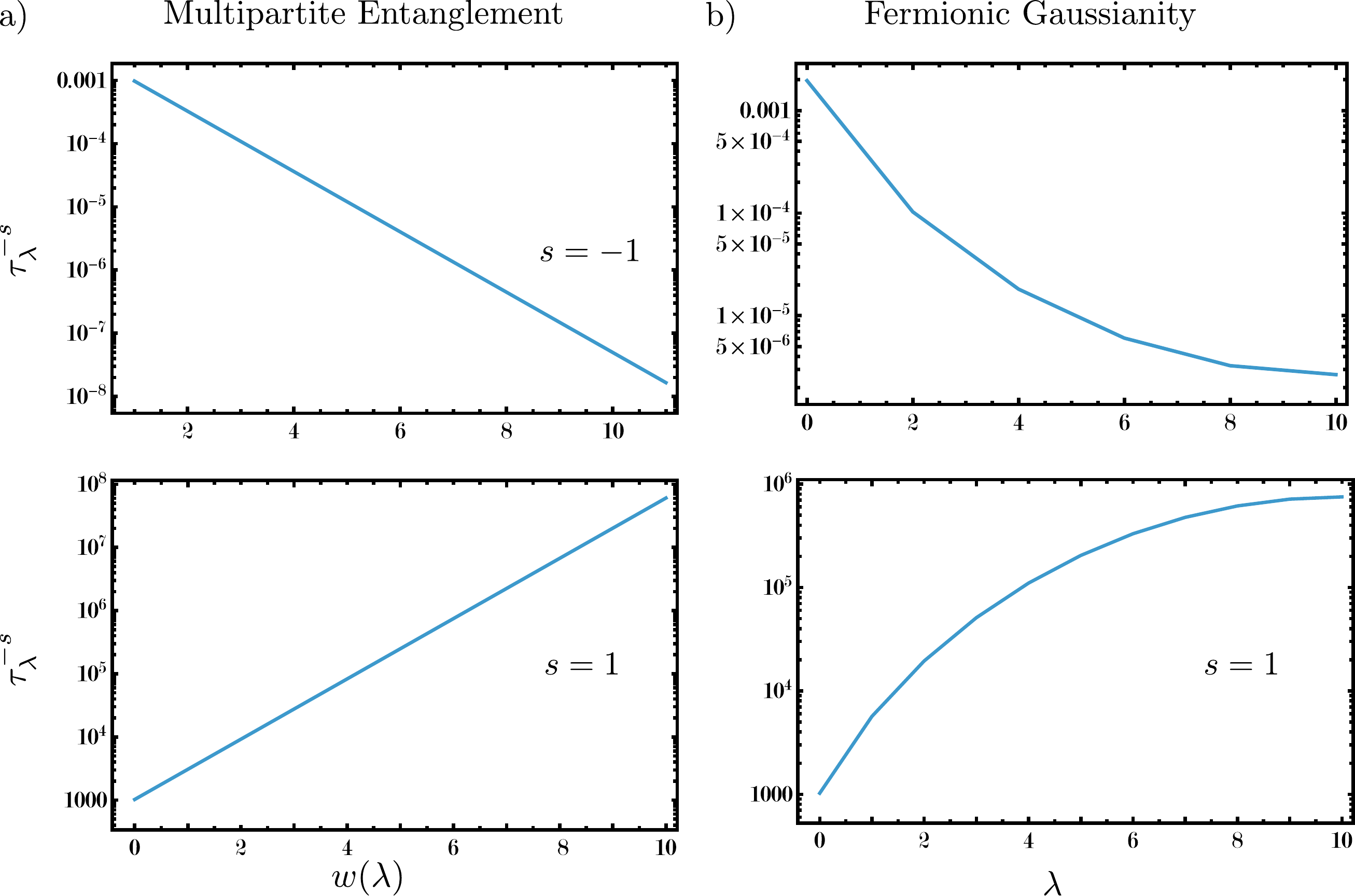}
    \caption{\textbf{Power multipliers $\tau_\lambda^{-s}$ for multipartite entanglement and fermionic Gaussianity at $n=10$.} (a) The multipartite-entanglement multipliers are shown against Hamming weight $w(\lambda)=0,\ldots,10$. (b) The fermionic multipliers use $\LC(\HC_+)$ with $d_+=2^9$, and the distinct even degrees $\lambda=0,2,\ldots,10$. Complementary degrees are identified; at the middle degree only the chirality-selected Hodge half remains. The upper row is $s=-1$ and the lower row is $s=1$. Both panels have been recomputed with the stated probability measure and actual irrep labels.}
    \label{fig:tau-others}
\end{figure}

\subsubsection{Fermionic Gaussianity}

Next, we focus on the QRT of
fermionic Gaussianity in a system of $n$ spinless fermions~\cite{gigena2015entanglement,weedbrook2012gaussian,hebenstreit2019all,diaz2023showcasing,jozsa2008matchgates,mele2024efficient,dias2023classical,goh2023lie,brod2011extending}. Here, free operations are obtained by unitary evolutions generated by quadratic Hamiltonians (in terms of Dirac's creation and annihilation operators). These unitaries, also known as matchgates, constitute a  restricted model of quantum computing~\cite{valiant2001quantum,knill2001fermionic,terhal2002classical,divincenzo2005fermionic}, and the addition of non-Gaussian states, or equivalently of non-matchgate unitaries, can promote these evolutions  to  universal quantum computation~\cite{lloyd1999quantum, knill2001scheme, bartlett2002universal, menicucci2006universal, jozsa2008matchgates, 
ohliger2010limitations,brod2011extending,brod2014computational, oszmaniec2017universal, zhuang2018resource}. There are $n$ complex fermionic modes and $2n$ Hermitian Majorana operators, represented on the $n$-qubit Fock space by
\begin{equation}
c_{2j-1}=Z_1\cdots Z_{j-1}X_j,\qquad
c_{2j}=Z_1\cdots Z_{j-1}Y_j,
\end{equation}
for $j=1,\ldots,n$. Quadratic unitaries preserve the parity operator $\Pi=Z_1\cdots Z_n$. We therefore formulate this example on the even-parity Hilbert space $\HC_+=\ker(\Pi-I)$, which contains the vacuum and has dimension $d_+=2^{n-1}$. For $n\geq2$, the free operations are the half-spin representation of $\operatorname{Spin}(2n)$ on $\HC_+$; their conjugation action factors through $\operatorname{SO}(2n)$. All traces and operator norms in this example refer to $\LC(\HC_+)$. The case $n=1$ is the one-dimensional, trivial example.

On $\HC_+$, products of $\lambda$ and $2n-\lambda$ Majoranas agree up to a phase because their product is proportional to parity. Thus the distinct sectors are the even degrees $0\leq\lambda\leq n$. If $V\simeq\mathbb C^{2n}$ is the vector representation, the multiplicity-free decomposition is
\begin{equation}\label{eq:irrep-L-ferm}
\LC(\HC_+)\simeq
\bigoplus_{\substack{0\leq\lambda<n\\\lambda\ {\rm even}}}\Lambda^\lambda V
\ \oplus\ \begin{cases}\Lambda^n_\chi V,&n\ {\rm even},\\0,&n\ {\rm odd}.
\end{cases}
\end{equation}
Here $\Lambda^n_\chi V$ denotes the one Hodge half selected by the chosen chirality. Write $\epsilon_\lambda=1$, except that $\epsilon_n=1/2$ when $n$ is even. The retained sectors have
\begin{equation}\label{eq:fermion-purities}
d_\lambda=\epsilon_\lambda\binom{2n}{\lambda},\qquad
\PC_\lambda(\dya{0^n})=\frac{\epsilon_\lambda}{d_+}\binom n{\lambda/2}.
\end{equation}
For $\lambda<n$, the vacuum overlaps with the $\binom n{\lambda/2}$ normalized products of $\lambda/2$ Pauli $Z$ operators. At $\lambda=n$, complementary products are identical on $\HC_+$, giving the factor $1/2$. These formulas satisfy $\sum_\lambda d_\lambda=d_+^2$ and $\sum_\lambda\PC_\lambda(\dya{0^n})=1$. Consequently, we find
\begin{equation}\label{eq:tau-fermion}
\tau_\lambda=\frac{\binom n{\lambda/2}}{2^{n-1}\binom{2n}{\lambda}}.
\end{equation}

Figure~\ref{fig:tau-others}(b) shows these multipliers. Frequency is ordered by the Casimir eigenvalue proportional to $\lambda(2n-\lambda)$, which increases over the retained degrees. The ratio $\tau_{\lambda+2}/\tau_\lambda=(\lambda+1)/(2n-\lambda-1)<1$ for successive retained degrees proves the low-pass behavior at $s=-1$ and the high-pass behavior at $s=1$.

To construct the QPS, we note that the subgroup of unitaries $\mathbb{U}(n)\subseteq\mathbb{G}$ composed of particle-preserving unitaries such that $\forall g \in \mathbb{U}(n)$ then $[T(g),\sum_i Z_i]=0$, stabilizes the all-zero state. Hence, the QPS is~\cite{dangniam2018quantum}
\begin{equation}
   \simeq\operatorname{Spin}(2n)/\widetilde{\operatorname U(n)}
\simeq\operatorname{SO}(2n)/\operatorname U(n).
\end{equation}
Now, while a direct parametrization of the group is possible, its form is not compact nor  straightforward. We refer the reader to~\cite{braccia2025optimal} for additional details on the Haar measure for matchgates.

To write the kernel without double counting complementary products, let $\mathcal J_\lambda$ contain every subset $J\subseteq\{1,\ldots,n\}$ of size $\lambda/2$ when $\lambda<n$, and one representative of each complementary pair when $\lambda=n$. With $Z_J=\prod_{j\in J}Z_j$ restricted to $\HC_+$,
\begin{align}
&\Delta(e,s)=\frac1{d_+}\sum_{\substack{0\leq\lambda\leq n\\\lambda\ {\rm even}}}
\tau_\lambda^{-(s+1)/2}\sum_{J\in\mathcal J_\lambda}Z_J,
\nonumber\\
&\Delta(\Omega,s)=T_+(\Omega)\Delta(e,s)T_+^\dagger(\Omega).
\end{align}
At $s=-1$ this recovers the vacuum projector and its coherent-state orbit.

\section{Generalization beyond SW kernels}\label{sec:general-kernels}
The preceding results concern the chosen $s$-parameterized family. However, covariance also permits more general operator-valued kernels, described by covariant integral quantization and frame theory~\cite{ali2000coherent,bergeron2014integral}. Importantly, in this more general framework, kernels need not even be positive. Changing them amounts to changing the fiducial operator, or its harmonic weights.

For definiteness, retain the same finite-dimensional harmonic image as the $s=0$ map and use its symmetric normalization. Let $A\mapsto\hat F_A$ and its dual $A\mapsto\check F_A$ be linear, $*$-preserving, covariant maps satisfying
\begin{align}
&\int_\XC\hat F_A(\Omega)d\mu(\Omega)
=\int_\XC\check F_A(\Omega)d\mu(\Omega)=\frac{\Tr A}{\sqrt d},\\
&\int_\XC\hat F_{A^\dagger}(\Omega)\check F_B(\Omega)d\mu(\Omega)
=\Tr[A^\dagger B].
\end{align}
More generally, distinct standardization constants $c$ and $\check c$ are possible if $c\check c=1/d$. Riesz representation gives kernels $\hat\Delta$ and $\check\Delta$ for these maps.

Reconstruction using the Wigner kernel gives
\begin{align}
\hat F_A(\Omega)&=\int_\XC d\mu(\Omega')\kappa(\Omega,\Omega')F_A(\Omega',0),\\
\kappa(\Omega,\Omega')&=\Tr[\hat\Delta(\Omega)\Delta(\Omega',0)].
\end{align}
The real kernel $\kappa$ is invariant under simultaneous translations of its arguments. Indeed, writing $\Omega=g\mathbb K$ and $\Omega'=g'\mathbb K$, define $k(g)=\kappa(g\mathbb K,e\mathbb K)$. This is a bi-$\mathbb K$-invariant function on $\mathbb G$, and
\begin{equation}
\hat F_A(g\mathbb K)=\int_{\mathbb G}dg'\,k(g'^{-1}g)F_A(g'\mathbb K,0),
\end{equation}
where $dg'$ is normalized Haar measure. Equivalently, this is an integral on $\mathbb G/\mathbb K$. 

When working with multiplicity-free cases, a covariant transformation acts by one scalar $k_\lambda$ on each irrep. In real matched bases,
\begin{align}
&\hat\Delta(\Omega)=\sum_{\lambda,j}k_\lambda Y_j^\lambda(\Omega)D_j^\lambda,
\nonumber\\
&\check\Delta(\Omega)=\sum_{\lambda,j}k_\lambda^{-1}Y_j^\lambda(\Omega)D_j^\lambda.
\end{align}
Then, the reality and invertibility conditions require real, nonzero $k_\lambda$, while symmetric standardization fixes $k_0=1$, equivalently $\int k=1$. These multipliers describe additional choices of fiducial operator beyond the SW ansatz. For general multiplicities, scalar $k_\lambda$ must be replaced by invertible matrices $K_\lambda$ on the multiplicity spaces, with dual blocks $(K_\lambda^\dagger)^{-1}$.

\section{Discussions}
In this work we have shown that the connection between QPSs and QRTs can be understood through the group Fourier decomposition of quantum states. In particular, the GFD purities of the reference free state determine the SW kernels and hence the filters applied to the irrep components of every state. This gives a general construction of QPS filters from the QRT's free states, and explains how the same group-theoretic structure organizes resources in both operator space and phase space.

From this perspective, the Cahill-Glauber parameter $s$ acquires a direct resource-theoretic meaning. Our examples of spin coherence, multipartite entanglement, and fermionic Gaussianity show how varying $s$ interpolates between free-adapted, low-pass representations at $s=-1$, an unchanged GFD spectrum at $s=0$, and high-pass representations at $s=1$. The smooth distributions obtained with the Husimi representation and the finer structure revealed by the $P$ representation can therefore be understood in terms of the relative weights assigned to different irreps. Changing $s$ allows us to examine different components of the same state and to understand which features of its phase-space representation arise from those components.

The $s$-duality in Proposition~\ref{prop:duality} further connects the phase-space descriptions of free and typical states. For nontrivial irreps, the average phase-space purity of a Haar-random state at parameter $s$ is proportional to that of the reference free state at $s+1$. Thus, the free-state GFD purities also determine the average phase-space spectra of typical states, which are highly resourceful in most of the QRTs considered. This provides an explicit relation between their otherwise different resource profiles through a shift within the same family of QPS representations.

Our results also give a way to understand features such as phase-space negativity through the irrep components that produce them. Indeed, in the spin-coherence example, even a free state has higher-spin components whose oscillatory harmonics can produce Wigner negativity. Changing the representation changes the relative contribution of these harmonics, while the state remains free. The group Fourier description therefore allows us to trace such phase-space features back to both the state's irrep content and the filtering imposed by the SW kernel.

Finally, the recent use of QPSs to describe quantum computations~\cite{heightman2025quantum} and study quantum kernel methods~\cite{chabaud2024phase} suggests further applications of this connection. In these settings, our results allow one to compare phase-space representations according to how they weight the irreps associated with a given QRT. A natural next step is to investigate whether this freedom can improve approximations of circuit dynamics or quantum kernels by identifying which components must be retained for a particular task. It would also be interesting to study how other resource quantifiers can be expressed in phase space, and to extend the connection between QPSs and QRTs to noncompact groups.
t.

\section*{Acknowledgements}
We thank Chris Ferrie, Ninnat Dangniam, Nathan Killoran and Blas Manuel Rodriguez Lara for useful discussions. L.C. was supported by the U.S. Department of Energy (DOE) through a quantum computing program sponsored by the Los Alamos National Laboratory (LANL) Information Science \& Technology Institute. L.C. acknowledges additional support by the National Science Foundation Graduate Research Fellowship under Grant No. 2140743 as well as the QSE G1 Fellowship at Harvard. N.L.D. was supported by the Center for Nonlinear Studies at LANL. M.L. and M.C. acknowledge support by the Laboratory Directed Research and Development (LDRD) program of LANL under project number 20260043DR. This work was also supported by the Quantum Science Center (QSC), a National Quantum Information Science Research Center of the U.S. DOE. Any opinions, findings, and conclusions or recommendations expressed in this material are those of the author(s) and do not necessarily reflect the views of the National Science Foundation.

\bibliography{quantum}

\appendix
\setcounter{lemma}{0}
\setcounter{theorem}{0}
\setcounter{proposition}{1}
\renewcommand{\tablename}{Appendix Table}

\clearpage
\onecolumngrid
\appendix

\section{Hilbert–space vs phase–space dictionary}
\label{app:dictionary}

For convenience, we summarize in Appendix Table~\ref{tab:dictionary} the main correspondences between the Hilbert space and phase space formalisms used in the paper.

\begin{table}[H]
    \centering
    \begin{adjustbox}{max width=\linewidth}
\begin{tabular}{|p{0.2\linewidth}|p{0.32\linewidth}|p{0.45\linewidth}|}
        \hline
        Concept & Hilbert space / operator space & Phase space / $L^2(\XC)$ \\
        \hline\hline
        State / operator &
        $\rho \in \LC(\HC)$, $A \in \LC(\HC)$ &
        $F_\rho(\Omega,s) = \Tr[\rho\,\Delta(\Omega,s)]$,
        $F_A(\Omega,s) = \Tr[A\,\Delta(\Omega,s)]$ \\
        \hline
        Inner product &
        Hilbert--Schmidt:
        $\langle A,B\rangle_{\mathrm{HS}}
        = \Tr[A^\dagger B]$ &
        $L^2$ inner product:
        $\langle f,g\rangle_{L^2(\XC)}
        = \int_\XC d\mu(\Omega)\, f^*(\Omega)\,g(\Omega)$,
        chosen so that
        $\langle A,B\rangle_{\mathrm{HS}}
        =\langle F_A(\cdot,s),F_B(\cdot,-s)\rangle_{L^2}$ \\
        \hline
        Group action &
        Adjoint action:
        $A \mapsto T(g) A T^\dagger(g)$ &
        Covariant action:
        $F_A(\Omega,s) \mapsto F_A(g^{-1}\!\cdot\!\Omega,s)$ \\
        \hline
        Irrep decomposition &
        $\LC(\HC) = \bigoplus_\lambda V_\lambda$ &
        $L^2(\XC) = \bigoplus_\sigma W_\sigma$ \\
        \hline
        Irrep bases &
        $\{D_j^\lambda\}_{j=1}^{d_\lambda}
        \subset V_\lambda$ orthonormal in HS &
        $\{Y_j^\sigma\}_{j=1}^{d_\sigma}
        \subset W_\sigma$ orthonormal in $L^2(\XC)$ \\
        \hline
        Mode coefficients &
        $A = \sum_{\lambda,j} a_j^\lambda D_j^\lambda$ &
        $F_A(\Omega,s)
        = \sum_{\sigma,j} \tilde a_j^\sigma(s) Y_j^\sigma(\Omega)$ \\
        \hline
        GFD purities &
        $P_\lambda(\rho)
        = \sum_j |\langle D_j^\lambda,\rho\rangle_{\mathrm{HS}}|^2$ &
        $\widetilde P_\lambda(F_\rho(\ ,s))
        = \sum_j |\langle Y_j^\lambda, F_\rho(\cdot,s)\rangle_{L^2}|^2$ \\
        \hline
        Relation between purities &
        --- &
        $\widetilde P_\lambda(F_\rho(\Omega,s))
        = \tau_\lambda^{-s} P_\lambda(\rho)$ \\
        \hline
        Operator multiplication &
        $AB$ &
        Twisted (star) product:
        $F_{AB}(\Omega,s) = (F_A \star_s F_B)(\Omega)$ \\
        \hline
    \end{tabular}
    \end{adjustbox}
    \caption{Schematic dictionary between the Hilbert--space/operator
    description and the quantum phase-space description at fixed $s$.}
    \label{tab:dictionary}
\end{table}

\section{Proof of Proposition~\ref{prop:irreps-Delta-LH}}

Here we present a proof of  Proposition~\ref{prop:irreps-Delta-LH}, which we recall for convenience. 

\begin{proposition}\label{prop:irreps-Delta-LH-AP}
    Let $\Delta(\Omega,s)$ be an SW kernel as defined in Eq.~\eqref{eq:SW-kernel}, then its GFD purities in the irreps of $\LC(\HC)$ are
\begin{equation}\label{eq:purity-SW-kernel-AP}
    \PC_\lambda(\Delta(\Omega,s))=\tau_\lambda^{-s} d_\lambda=\left(\frac{\mathcal{P}_\lambda(\dya{\rm hw})}{d_\lambda}\right)^{-s}d_\lambda\,.
\end{equation}
\end{proposition} 
\begin{proof}
    We begin by recalling from Eq.~\eqref{eq:SW-kernel} the definition of the SW kernel,
\begin{equation}   \label{eq:SW-kernel-AP}
\Delta(\Omega,s)=\sum_\lambda\sum_{j=1}^{d_\lambda}\tau_\lambda^{-s/2}Y^\lambda_j(\Omega)D_j^\lambda\,,
\end{equation}
from where we find the  projection onto the $\lambda$
\begin{equation}   \Delta_\lambda(\Omega,s)=\sum_{j=1}^{d_\lambda}\tau_\lambda^{-s/2}Y^\lambda_j(\Omega)D_j^\lambda\,.
\end{equation}
Next, combining the covariance property of the kernel in Definition~\ref{def:SW-kernel}, along with the invariance of the GFD purities of Eq.~\eqref{eq:purity-inv} we find 
\begin{equation}
\PC_\lambda(\Delta(\Omega,s))=\PC_\lambda(T(\Omega)\Delta(e,s)T\ad(\Omega))=\PC_\lambda(\Delta(e,s))\,,
\end{equation}
which gives
\begin{equation}
    \PC_\lambda(\Delta(\Omega,s))=\tau_\lambda^{-s}\sum_{j=1}^{d_\lambda}|Y^\lambda_j(e)|^2\,.
\end{equation}

By covariance, $\sum_j|Y_j^\lambda(\Omega)|^2$ is constant on the homogeneous space. Integrating it with the probability measure and using orthonormality gives
\begin{equation}\label{eq:explicit}
\sum_{j=1}^{d_\lambda}|Y_j^\lambda(\Omega)|^2
=\int_\XC d\mu(\Omega')\sum_{j=1}^{d_\lambda}|Y_j^\lambda(\Omega')|^2=d_\lambda.
\end{equation}
Consequently, $\PC_\lambda(\Delta(\Omega,s))=\tau_\lambda^{-s}d_\lambda$, as claimed.

\end{proof}

\section{Proof of Proposition~\ref{prop:irreps-rho-LX}}

Let us now present a derivation of Proposition~\ref{prop:irreps-rho-LX}, which we now restate:

\begin{proposition}\label{prop:irreps-rho-LX-Ap}
    Let $\rho$ be a quantum state and $F_{\rho}(\Omega,s)$ its phase space representation. Then, from Definition~\ref{def:GFD-purities-LX}, the irrep projections are non-zero only for the irreps $\sigma=\lambda$ that appear in the decomposition of $\LC(\HC)$ in Definition~\ref{def:GFD-purities-LH}. In this case we have
    \begin{equation}
        [F_{\rho}]_\lambda(\Omega,s)
        = \sum_{j=1}^{d_\lambda}
          \tau_\lambda^{-s/2}\,Y^\lambda_j(\Omega)\,
          \big\langle D_j^\lambda,\rho\big\rangle_{\LC}\,,
    \end{equation}
    and the ensuing purities
    \begin{align}
        \widetilde{\PC}_\lambda\big(F_{\rho}(\Omega,s)\big)
        &= \tau_\lambda^{-s}\,\PC_\lambda(\rho)
         = \left(\frac{\mathcal{P}_\lambda(\dya{{\rm hw}})}{d_\lambda}\right)^{-s}
           \PC_\lambda(\rho)\,.
    \end{align}
\end{proposition}

\begin{proof}
    We start by combining the definition of the SW kernel in Eq.~\eqref{eq:SW-kernel} (or Eq.~\eqref{eq:SW-kernel-AP}) and of the function $F_\rho$ 
    \begin{equation}
F_\rho(\Omega,s)=\Tr[\rho\Delta(\Omega,s)]\,.
\end{equation}
to obtain 
\begin{equation}   F_\rho(\Omega,s)=\sum_\lambda\sum_{j=1}^{d_\lambda}\tau_\lambda^{-s/2}Y^\lambda_j(\Omega)\langle \rho,D_j^\lambda\rangle_\LC\,.
\end{equation}

Projecting onto a harmonic sector gives
\begin{equation}
[F_\rho]_\sigma(\Omega,s)=
\begin{cases}
\displaystyle\sum_{j=1}^{d_\sigma}\tau_\sigma^{-s/2}Y_j^\sigma(\Omega)
\langle D_j^\sigma,\rho\rangle_\LC,&\sigma\text{ occurs in }\LC(\HC),\\
0,&\text{otherwise}.
\end{cases}
\end{equation}

Combining the previous with Definition~\ref{def:GFD-purities-LX}, we obtain that the non-zero purities are\footnote{To indicate that the purities are non-zero, we will use the notation $\sigma=\lambda$.}
\begin{equation}
    \widetilde{\PC}_\lambda(F_{\rho}(\Omega,s))=\sum_{j=1}^{d_\lambda}\tau_\lambda^{-s}\left\langle D_j^\lambda,\rho\right\rangle_{\LC}^2=\tau_\lambda^{-s}\PC_\lambda(\rho)\,,
\end{equation}
where we used the fact that $\langle \rho,D_j^\lambda\rangle_\LC=\langle D_j^\lambda,\rho\rangle_\LC\in\mathbb{R}$ as the bases of $\LC(\HC)$ are taken to be Hermitian. 
\end{proof}

\end{document}